%% file: Griefing_arxiv.tex
\documentclass[a4paper,11pt]{article}
\usepackage[table]{xcolor}
\usepackage[utf8]{inputenc}
\usepackage{amsmath,amssymb,amsthm,graphicx,pgfplots,enumitem,comment,csquotes,subcaption,doi,adjustbox,multirow}
\pgfplotsset{compat=newest}
\usepgfplotslibrary{dateplot}
\usetikzlibrary{cd}
\usepackage{colortbl,xcolor,booktabs}
\usepackage{algorithm}
\usepackage{tikz,latexsym}
\usepackage[noend]{algpseudocode}
\usepackage{authblk,booktabs,tabularx}
\usepackage[margin=1in]{geometry}

\allowdisplaybreaks
\let\oldbibliography\thebibliography
\renewcommand{\thebibliography}[1]{\oldbibliography{#1}
\setlength{\itemsep}{1pt}}

\usepackage{hyperref,cleveref}

\input{header}
\newcommand{\lt}{\left[}
\newcommand{\rt}{\right]}
\newcommand{\qt}{\enquote}

\newcommand{\x}{x_{ij}}

\newcommand{\rhoi}{{\rho_i}}

\newcommand{\GF}{\operatorname{GF}}
\newcommand{\xe}{x^{\text{ESA}}}

\newcommand{\gam}{\(N,(u_i,c_i)_{i\in N}\)}
\newcommand{\gamc}{\(N,c,u_{i\in N}\)}

\renewcommand{\Pi}{u}
\renewcommand{\(}{\left(}
\renewcommand{\)}{\right)}
\renewcommand{\b}{b_{ij}}
\renewcommand{\v}{v_{ij}}

\renewcommand{\comment}[1]{\ignorespaces}

\renewcommand{\b}{b_{ij}}
\renewcommand{\v}{v_{ij}}
\renewcommand{\vec}[1]{\mathbf{#1}}

\newcolumntype{L}{l>{\hspace*{-0.6\tabcolsep}}}
\newcolumntype{R}{r<{\hspace*{-0.6\tabcolsep}}}
\newcolumntype{C}{c<{\hspace*{-0.6\tabcolsep}}}
\definecolor{ppcolor}{rgb}{0.212,0.506,0.584}
\definecolor{pcolor}{rgb}{0.255,0.596,0.686}
\definecolor{rcolor}{rgb}{0.294,0.675,0.776}
\definecolor{ecolor}{rgb}{0.569,0.765,0.835}
\definecolor{scolor}{rgb}{0.733,0.843,0.890}
\definecolor{ocolor}{rgb}{0.783,0.893,0.940}
\definecolor{oocolor}{rgb}{0.783,0.893,0.940}

\newtheorem{theorem}{Theorem}
\newtheorem{lemma}[theorem]{Lemma}
\newtheorem{corollary}[theorem]{Corollary}
\newtheorem{proposition}[theorem]{Proposition}
\newtheorem{observation}{Observation}
\theoremstyle{definition}
\newtheorem{definition}[theorem]{Definition}
\theoremstyle{remark}
\newtheorem*{remark}{Remark}
\Crefname{notat}{Notation}{Notations}
\Crefname{observation}{Observation}{Observations}

\makeatletter
\def\BState{\State\hskip-\ALG@thistlm}
\makeatother
\let\oldReturn\Return
\renewcommand{\Return}{\State\oldReturn}
\algnewcommand{\Initialize}[1]{%
  \State \textbf{Initialize:} \hspace*{0.2em}\parbox[t]{.75\linewidth}{\raggedright #1}
}
\usepackage{environ}
\makeatletter
\newsavebox{\measure@tikzpicture}
\NewEnviron{scaletikzpicturetowidth}[1]{%
  \def\tikz@width{#1}%
  \begin{lrbox}{\measure@tikzpicture}%
  \BODY
  \end{lrbox}%
  \pgfmathparse{#1/\wd\measure@tikzpicture}%
  \BODY
}
\makeatother

\title{From Griefing to Stability in Blockchain Mining Economies}

\author[1]{Yun Kuen Cheung}
\author[2]{Stefanos Leonardos}
\author[2]{Georgios Piliouras}
\author[2]{Shyam Sridhar}
\affil[1]{Royal Holloway University of London, \emph{yunkuen.cheung@rhul.ac.uk}}
\affil[2]{Singapore University of Technology and Design, \emph{\{stefanos\_leonardos, georgios\}@sutd.edu.sg}, \emph{shyam\_sridhar@mymail.sutd.edu.sg}}
\date{}

\begin{document}

\maketitle

\begin{abstract}
We study a game-theoretic model of blockchain mining economies and show that \emph{griefing}, a practice according to which participants harm other participants at some lesser cost to themselves, is a prevalent threat at its Nash equilibria. The proof relies on a generalization of evolutionary stability to non-homogeneous populations via \emph{griefing factors} (ratios that measure network losses relative to deviator's own losses) which leads to a formal theoretical argument for the dissipation of resources, consolidation of power and high entry barriers that are currently observed in practice. \par
A critical assumption in this type of analysis is that miners' decisions have significant influence in aggregate network outcomes (such as network hashrate). However, as networks grow larger, the miner's interaction more closely resembles a distributed production economy or \emph{Fisher market} and its stability properties change. In this case, we derive a \emph{proportional response (PR)} update protocol which converges to market equilibria at which griefing is irrelevant. Convergence holds for a wide range of miners risk profiles and various degrees of resource mobility between blockchains with different mining technologies. Our empirical findings in a case study with four mineable cryptocurrencies suggest that risk diversification, restricted mobility of resources (as enforced by different mining technologies) and network growth, all are contributing factors to the stability of the inherently volatile blockchain ecosystem.
\end{abstract}


\section{Introduction}\label{sec:introduction}

With more than 4000 circulating cryptocurrencies, currently valued above the staggering amount of \$1 trillion \cite{Reu21}, and countless other decentralized applications running on them \cite{Sta21}, the underlying blockchain technologies are attracting increasing attention. However, a (still persisting) barrier in their wider public adoption is the uncertainty regarding their stability and long-term sustainability. Understanding these factors is important both for the success of permissionless blockchains and for the acceptance of cryptocurrencies as a means for widespread monetary transactions \cite{Bud18,Aue19,Gan19}.\par
The critical actors for the stability of the blockchain ecosystem are the miners who provide their costly resources (e.g., computational power in Proof of Work (PoW) or units of the native cryptocurrency in Proof of Stake (PoS) protocols) to secure consensus on the growth of the blockchain \cite{Gar15,Bon15,Be16}. Miners act in a self-interested, decentralized manner and may enter or leave these networks at any time. For their service, miners receive monetary rewards in return, typically in the form of transaction fees and newly minted coins of the native cryptocurrency in proportion to their individual resources in the network \cite{Gar15,But19,Chen19}. \par
The total amount of these resources, their distribution among miners, and the consistency in which they are provided are fundamental factors for the reliability of all blockchain supported applications. However, despite their importance, miners' incentives to allocate and distribute their resources among various blockchains are yet far from understood. Existing studies \cite{Fia19,Gor19,Eas19,Jin20} and online resources that reflect investor and blockchain related sentiment \cite{Coi20,Cit21}, all provide compelling evidence that the \emph{allocation of mining resources in the blockchain economy} is a still largely under-explored area. 

\paragraph{Model and contribution} Motivated by the above, we study a game-theoretic model of the mining economy (comprising a single or multiple co-existing blockchains) and reason about miners' resource allocations. Our starting point is the work of \cite{Arn18} who derive the unique \emph{Nash Equilibrium} (NE) allocations under the proportional reward scheme that applies to most Proof of Work (PoW) and Proof of Stake (Pos) protocols (\Cref{thm:arnosti}). Our first observation is that at the predicted NE levels, active miners are still incentivised to deviate (by increasing their resources) in order to achieve higher \emph{relative payoffs}. While behaving sub-optimally in terms of their absolute payoffs, the loss that a deviating miner incurs to themselves is overcompensated by a larger market share and a higher loss that is incurred to each other individual miner and hence, to the rest of the network as a whole (\Cref{thm:grief}, \Cref{cor:single}). \par
This practice, in which participants of a network cause harm to other participants, even at some cost to themselves, is known as \emph{griefing}.\footnote{The term \emph{griefing} originated in multiplayer games \cite{Gri21} and was recently introduced in blockchain related settings by \cite{But17}.} Our main technical insight is that griefing is closely related to the game-theoretic notion of \emph{evolutionary stability}. Specifically, we quantify the effect of a miner's deviation via the \emph{(individual) Griefing Factors (GF)}, defined as the ratios of network (or individual) losses over the deviator's own losses (\Cref{def:gf}), and show that an allocation is \emph{evolutionary stable} if and only if all its individual GFs are less than $1$ (\Cref{lem:equivalent}).\footnote{An allocation has a griefing factor of $k$ if a miner can reduce others' payoffs by \$k at a cost of \$1 to themselves by deviating to some other allocation.} We call such allocations \emph{(individually) non-griefable} (\Cref{def:griefable}). This equivalence for homogeneous populations (i.e., for miners having equal mining costs) for which evolutionary stability is defined. However, as GFs are defined for arbitrary populations (not necessarily homogeneous), it provides a way to generalize evolutionary stable allocations as individually non-griefable allocations. Rephrased in this framework, our result states that the NE allocation is always individually griefable by miners who unilaterally increase their resources (\Cref{thm:grief}).\par
The previous evolutionary argument provides a theoretical explanation for the increasing dissipation of mining resources (above optimal levels) in PoW protocols and an alternative rationale for the concentration of mining power in few entities that is observed in both PoW and PoS protocols \cite{Arn18,Kwo19,Leo20}. A distinctive feature of this \emph{over-mining} behavior in comparison to in-protocol adversarial behavior, e.g., \cite{Kia16,Eya18,Alk19,Sin20}, is that it does not directly compromise the functionality of the blockchain. When a single miner increases their resources, the safety of the blockchain also increases. However, this practice has multiple negative byproducts as it generates a trend towards market concentration, dissipation of resources and high entry barriers.\par
With griefing being a concern for instability at the NE allocations (and increased dissipation of resources being a concern at the non-griefable or evolutionary stable allocations (\Cref{rem:part2}, \Cref{prop:breakeven}), it is not immediately clear how to extend this model to study allocation of resources in the general case of multiple co-existing blockchains. A critical observation is that these types of equilibria (both Nash and evolutionary stable) are derived in a model in which individual miners are assumed to influence aggregate market outcomes with their strategic decisions. However, despite the currently observed concentration of mining power, this assumption may not be satisfied in practice as mining networks continue to expand and is certainly not satisfied in the originally envisioned architecture in which (permissionless) mining networks were expected to function in a genuinely decentralized fashion \cite{Nak08}. Thus, the question that naturally arises is whether we can reason about the stability of the ecosystem under the assumption of negligible individual influence.\par
To address this question, we extend the initial model to the multiple blockchain setting under the assumption that each miner has a finite capacity of resources that is negligible in comparison to collective network levels (large market assumption \cite{Col16}). Under these conditions griefing becomes irrelevant. The ensuing model is mathematically equivalent to a \emph{Fisher market} (or production economy \cite{Bra18}) in which miners correspond to buyers, goods to revenues from different cryptocurrencies and prices to aggregate allocated resources (\Cref{sec:multiple}). We endow this model with \emph{quasi-constant elasticity of substitution (quasi-CES)} utilities (parameterized by a miner-specific substitution parameter $\rho_i$)  to account for risk diversification and various degrees of resource mobility between different blockchain technologies. \par
Our main theoretical contribution in this part is the derivation of a \emph{Proportional Response (PR)} protocol that converges to the market equilibria of the model for any quasi-CES utilities with substitution parameters $\rho_i\in[0,1]$ (\Cref{thm:equilibrium}). The protocol requires as inputs only local (i.e., miner specific) and collective (total revenues and resources, e.g., estimated hashrate) information which make it particularly suitable for such distributed economies from a practical perspective (\Cref{alg:proportional}).\footnote{While our techniques to derive the (PR) protocol and establish its convergence are based on well-known approaches, the result is novel and may be of independent theoretical interest in the study of exchange \cite{Zha11,Che18} or distributed production economies \cite{Bra18}.} By contrast, we show that learning protocols that are commonly used in game-theoretic settings, such as Gradient Ascent (GA) and Best Response (BR) dynamics, exhibit chaotic or highly irregular behavior (\Cref{sub:comparison}). This is true even for large number of miners and, in the case of GA, even for relative small step-size (as long as miners are assumed to have influence via non-binding capacities on collective outcomes). Interestingly, these findings establish another source of instability of the equilibria of the game-theoretic model that is different in nature from the previous ones (algorithmic versus incentive driven). \par
\paragraph{Case Study:} We use the (PR) dynamics to study the equilibrium allocations of a representative miner in a blockchain economy with four popular cryptocurrencies: Bitcoin, Bitcoin Cash, Ethereum and Litecoin (\Cref{sec:case}). Our empirical results, that are based on empirical data (daily revenues and hashrates over a period of three years), suggest that the \emph{Proportional Profitability Ratio (PPR)}, which is defined as the normalized ratio of revenue over expenses for each coin (\Cref{def:ppr}), is an important metric to understand miner's behavior in the blockchain economy. Specifically, risk neutral miners  (equivalently, miners with full mobility of resources) allocate all their resources to the coin with the highest PPR. However, miners with intermediate values of risk neutrality (restricted mobility of resources), distributes their resources precisely in proportion to the PPR of the four available coins (\Cref{fig:abstract}). Our findings suggest that restricted mobility of resources (as enforced by the use of different mining technologies in the various blockchains), risk diversification and growth of the mining networks, are all factors that contribute to the stability of the emerging blockchain ecosystem.

\paragraph{Other Related Works}\label{sub:related} Our paper main contributes to the growing literature on miners' incentives in blockchain networks. The two derived sources of instabilities, griefing and fluctuating allocations derived by greedy update rules, complement existing results concerning inherent protocol instabilities \cite{Bo15,Car16,Dim17}, manipulation of the difficulty adjustment in PoW protocols \cite{Fia19,Gor19,Shu20} or adversarial behavior \cite{Kia16,Eya18,Bro19,Alk19}. Our findings in the case of a single blockchain support the accumulating evidence that decentralization is threatened in permissionless blockchains \cite{Arn18,Kwo19,Leo20} and offer an (evolutionary) explanation for the increased dissipation of resources (above optimal levels) that is observed in the main PoW mining networks \cite{Vri18,Sto19,Vri20}. Our market model approach, provides the first (to our knowledge) modeling and equilibrium analysis of the blockchain mining economy as a whole (multiple co-existing blockchains) and contributes to the related literature that is still under early development \cite{Spi18,Bis19,Jin20}.\par
Technically, our models mirror the model of single or multiple simultaneous Tullock contests (equivalently all-pay auctions) and the model of Fisher markets with quasi-CES utilities. Thus, some elements of the paper, in particular the notion of griefing factors and the convergence of the PR dynamics, may be of independent interest in the study of general evolutionary, game-theoretic models (in arbitrary non-homogeneous populations) for decentralized markets \cite{LLSB08,Dip09,Hor10}, and in distributed production economies, respectively \cite{Zha11,Cheu19,Dvi20}.

\paragraph{Outline}\label{sub:outline} \Cref{sec:single} presents the strategic model (in a single blockchain) and studies the notions of griefing and evolutionary stability. \Cref{sec:multiple} comprises the market model (with multiple blockchains and negligible individual capacities), the proportional response protocol and a comparison between the two models (\Cref{sub:comparison}). \Cref{sec:case} contains the empirical results and \Cref{sec:conclusions} concludes the paper. All proofs of \Cref{sec:single,sec:multiple} are deferred to \Cref{app:omitted,app:proportional}, respectively.

\section{Allocation of Mining Resources: Strategic Model}\label{sec:single}

For the first part of our analysis, we will study mining in a single blockchain. We will introduce some additional notation in \Cref{sec:multiple}, when we will study the allocation of mining resources in multiple blockchains.

\subsection{Model and Nash Equilibrium Allocations}\label{sec:model}
We consider a network of $N=\{1,2,\dots,n\}$ miners who allocate their resources, $x_i\ge0$, to mine a blockchain-based cryptocurrency. Each miner $i\in N$ has an individual per unit cost $c_i>0$. For instance, in Proof of Work (PoW) mining, $x_i$ corresponds to TeraHashes per second (TH/s) and $c_i$ to the associated costs (energy, amortized cost of hardware etc.) of producing a TH/s. We will write $\vec{x}=\(x_i\)_{i\in N}$ to denote the vector of allocated resources of all miners, and $X=\sum_{i=1}^nx_i$ to denote their sum. We will also write $v$ to denote the total miners' revenue (coinbase transaction reward plus transaction fees) in a fixed time period (typically an epoch or a day in the current paper). The market share of each miner is proportional to their allocated resources (as is the case in most popular cryptocurrencies, see e.g., \cite{Nak08,But17}). Thus, the utility of each miner is equal to 
\begin{equation}\label{eq:utility}
\Pi_i\(x_i,\vec{x}_{-i}\)=\frac{x_i}{x_i+X_{-i}}v-c_ix_i, \qquad \text{for all } i\in N,
\end{equation}
where, following standard conventions, we write $\vec{x}_{-i}=\(x_j\)_{j\neq i}$ and $X_{-i}:=\sum_{j\neq i} x_j$ to denote the vector and the sum, respectively, of the allocated resources of all miners other than $i$. In equation \eqref{eq:utility}, we may normalize $v$ to $1$ without loss of generality (by scaling each miner's utility by $v$). We will refer to the game, $\Gamma=\gam$, defined by the set of miners $N$, the utility functions $\Pi_i, i\in N$ and the cost parameters $c_i, i\in N$ as the \emph{mining game $\Gamma$}. As usual, a \emph{Nash equilibrium} is a vector $\vec{x}^*$ of allocations $x_i^*, i\in N$, such that  
\begin{equation}\label{eq:nash}
\Pi_i\(\vec{x}^*\)\ge\Pi_i\(x_i,\vec{x}_{-i}^*\), \qquad \text{for all } x_i\neq x_i^*, \text{ for all miners } i\in N.
\end{equation}
In terms of its Nash equilibrium, this game has been analyzed by \cite{Arn18}. To formulate the equilibrium result, let 
\begin{equation}\label{eq:cstar}
c^*:=\frac1{n-1}\sum_{i=1}^nc_i,
\end{equation} and assume for simplicity that $c^*>c_i$ for all $i\in N$. This is a \emph{participation constraint} and implies that we consider only miners that are active in equilibrium. The unique Nash equilibrium of $\Gamma$ is given in \Cref{thm:arnosti}.

\begin{theorem}[\cite{Arn18}]\label{thm:arnosti}
At the unique pure strategy Nash equilibrium of the mining game $\Gamma$, miner $i\in N$ allocates resources $x_i^*=\(1-c_i/c^*\)/c^*$. In particular, the total mining resources, $X^*$, allocated at equilibrium are equal to $X^*=1/c^*$.
\end{theorem}

\Cref{thm:arnosti} is our starting point. Our first task is to test the robustness of this Nash equilibrium in the context of decentralized and potentially adversarial networks. For instance, while the Nash equilibrium outcome is well-known to be incentive compatible, an adversary may decide to harm others by incurring a low(er) cost to himself. In decentralized networks, the (adversarial) practice of harming others at some lesser own loss is termed \emph{griefing} \cite{But17}. As we show next, griefing is indeed possible in this case: a miner who increases their allocated resources above the Nash equilibrium prediction forgoes some of their own profits but incurs a (considerably) larger loss to the rest of the network. Our proof exploits a link between griefing and the fact that the Nash equilibrium is not evolutionary stable. To make these statements explicit, we first provide the relevant framework.

\subsection{Evolutionary Stable Allocations and Griefing Factors}\label{sub:esa}

For this part, we restrict attention to homogeneous populations of miners, for which the notion of \emph{evolutionary stability} is defined. Specifically, we consider a mining game $\Gamma=\gam$ such that all miners have equal costs, i.e., $c_i=c$ for some $c>0$, for all $i\in N$. We will write $\Gamma = \gamc$ and we will call this mining game \emph{symmetric}. In this case, $c^*=\frac{n}{n-1}c$ and each miner allocates $x_i^*=\frac{n-1}{n^2c}$ resources in the unique (symmetric) pure strategy Nash equilibrium of $\Gamma$. The symmetry assumption implies that $u_i(\vec{x})= u_j\(\vec{x}\)$ for all $i,j\in N$ and for any allocation $\vec{x}=\(x_i\)_{i\in N}$. The following definition of evolutionary stability due to \cite{Sch88, Heh04} requires the weaker condition that $u_i(\vec{x})= u_j\(\vec{x}\)$ for all $i,j\in N$ and for any \emph{symmetric} allocation $\vec{x}=\(x_i\)_{i\in N}$. In the case of the utility functions in equation \eqref{eq:utility}, these two conditions are equivalent.
\begin{definition}[Evolutionary Stable Allocation (ESA), \cite{Sch88,Heh04}]\label{def:esa}
Let $\Gamma=\gam$ be a mining game such that $\Pi_i\equiv \Pi_j$ for all $i,j\in N$ for all symmetric allocation profiles $\vec{x}\ge0$. Then, a symmetric vector $\vec{\xe}=\(\xe\)_{i\in N}$ is an \emph{evolutionary stable allocation (ESA)} if
\begin{equation}\label{eq:esa}
\Pi_i\(x_i,\vec{\xe_{-i}}\)<\Pi_j\(x_i,\vec{\xe_{-i}}\), \quad \text{ for all } j\neq i \in N, x_i\neq \xe.
\end{equation}
\end{definition}
\Cref{def:esa} implies that an ESA, $\vec{\xe}$, maximizes the relative payoff function, $\Pi_{i}(x_i,\vec{\xe_{-i}}) - \Pi_j(x_i,\vec{\xe_{-i}})$ with $j \in N, j\neq i$, of any miner $i\in N$. Intuitively, if all miners select an ESA, then there is no other allocation that could give an individually deviating miner a higher \emph{relative payoff}. In other words, if a symmetric allocation $x_i=x, i\in N$, is not ESA, then there exists a $x'\neq x$, so that a single miner who deviates to $x'$ has a strictly higher payoff (against $x$ of the other $n-1$ miners) than every other miner who allocates x (against $n-2$ other miners who allocate $x$ and the deviator who allocates $x'$) \cite{Heh04}. \par
As mentioned above, evolutionary stability is defined for homogeneous populations and may be, thus, of limited applicability for practical purposes. To study non homogeneous populations, we will need a proper generalization of evolutionary stability. To achieve this, we introduce the notion of \emph{griefing factors} which, as we show, can be used to formulate evolutionary stability and which is readily generalizable to arbitrary settings. This is done next.\footnote{In the current setting, the assumption of symmetric miners (miners with equal or at least almost equal cost) is less restrictive than it seems. The participation constraint $c_i<c^*=\frac{1}{n-1}\sum_{j=1}^nc_j$ implies that the costs, $c_i$'s, of the active miners in equilibrium cannot be \emph{too different}. This is formalized in \Cref{obs:bound} in \Cref{app:omitted}.}

\begin{definition}[Griefing Factors (GF)]\label{def:gf}
Let $\Gamma=\gam$ be a mining game (not necessarily symmetric) in which all miners are using the allocations $x_i^*, i\in N$, and suppose that a miner $i$ deviates to an allocation $x_i\neq x_i^*$. Then, the \emph{griefing factor, (GF), of strategy $x_i$ with respect to strategy $x^*$} is defined by 
\begin{align}\label{eq:grief}
\GF_i\(\(x_i,\vec{x}_{-i}^*\);\vec{x}^*\):=\,&\frac{\text{loss incurred to the network}}{\text{deviator's own loss}}=\frac{\sum_{j\neq i}^n\lt\Pi_j\(\vec{x}^*\)-\Pi_j\(x_i,\vec{x}_{-i}^*\)\rt}{\Pi_i\(\vec{x}^*\)-\Pi_i\(x_i,\vec{x}_{-i}^*\)}\,,
\intertext{for all $i\in N$, where \emph{loss} is the same as \emph{utility loss}. The $\GF$ with respect to an allocation $x^*$ can be then defined as the supremum over all possible deviations, i.e.,}
\GF\(\vec{x}^*\)=\,&\sup_{i\in N, x_i\ge0}\left\{\GF_i\(\(x_i,\vec{x}_{-i}^*\);\vec{x}^*\)\right\}.\nonumber
\intertext{We can also define the \emph{individual griefing factor of strategy $x_i$ with respect to strategy $x^*$ against a specific miner $j$}, as follows}
\label{eq:individual}
\GF_{ij}\(\(x_i,\vec{x^*_{-i}}\);\vec{x^*}\):=\,&\frac{\text{loss incurred to miner $j$}}{\text{deviator's own loss}}=\frac{\Pi_j\(\vec{x}^*\)-\Pi_j\(x_i,\vec{x}_{-i}^*\)}{\Pi_i\(\vec{x}^*\)-\Pi_i\(x_i,\vec{x}_{-i}^*\)}
\end{align}
for all $j\neq i \in N$, where as in equation \eqref{eq:grief}, \emph{loss} is a shorthand for \emph{utility loss}. It holds that $\GF_i\(\(x_i,\vec{x^*_{-i}}\);\vec{x^*}\)=\sum_{j\neq i}\GF_{ij}\(\(x_i,\vec{x^*_{-i}}\);\vec{x^*}\)$.
\end{definition}
As mentioned in \Cref{def:gf}, the numerator of $\GF$ corresponds to the loss of all miners other than $i$ incurred by $i$'s deviation to $x_i$, whereas the denominator corresponds to miner $i$'s own loss (cf. equation \eqref{eq:grief}). In decentralized mechanisms (e.g., blockchains), this metric captures an important \emph{incentive compatibility} condition: namely, a mechanism is safe against manipulation if the costs of an attack exceed its potential benefits to the attacker \cite{Bud18,Aue19,Gan19}. This motivates to define an allocation as \emph{griefable} if its GF is larger than $1$. 

\begin{definition}[Griefable and Individually Griefable Allocations]\label{def:griefable}
An allocation $\vec{x^*}=\(x^*_i\)_{i\in N}$ is \emph{griefable} if $\GF(\vec{x^*})>1$. An allocation $\vec{x^*}=\(x^*_i\)_{i\in N}$ is \emph{individually griefable} if there exist $i,j \in N$ and $x_i\neq x^*_i\ge0$, such that the individual griefing factor $\GF_{ij}\(\(x_i,\vec{x^*_{-i}}\);\vec{x^*}\)$ is larger than $1$.
\end{definition}

An important observation is that the condition of evolutionary stability can be expressed in terms of the individual griefing factors. In particular, an allocation $\vec{\xe}$ is evolutionary stable if and only if all \emph{individual} griefing factors are less than $1$, i.e., if and only if $\vec{\xe}$ is not individually griefable. This is formalized in \Cref{lem:equivalent}.
\begin{lemma}\label{lem:equivalent}
Let $\Gamma=\gamc$ be a symmetric mining game. Then, an allocation $\vec{\xe}=\(\xe\)_{i\in N}$ is evolutionary stable if and only if $\vec{\xe}$ is not griefable, i.e., iff
\begin{equation}\label{eq:char}
\GF_{ij}\(\(x_i,\vec{\xe_{-i}}\);\vec{\xe}\)<1, \quad \text{ for all } j\neq i \in N, x_i\neq \xe.
\end{equation}
\end{lemma}
\begin{proof}
Since $\Pi_i\equiv \Pi_j$ for all symmetric $\vec{x}$ and all $i,j\in N$ by assumption, we may write equations \eqref{eq:esa} as 
\[\Pi_i\(x',\vec{\xe_{-i}}\)-\Pi_{i}\(\vec{\xe}\)<\Pi_j\(x',\vec{\xe_{-i}}\)-\Pi_j\(\vec{\xe}\),\]
for all $j\neq i \in N$ and for all $x'\neq \xe$. Since $\Pi_i\(x',\vec{\xe_{-i}}\)<\Pi_i\(\vec{\xe}\)$ for all $x\neq \xe$ and for any miner $i\in N$, we may rewrite the previous equation as \[1 >\frac{\Pi_j\(x',\vec{\xe_{-i}}\)-\Pi_j\(\vec{\xe}\)}{\Pi_i\(x',\vec{\xe_{-i}}\)-\Pi_{i}\(\vec{\xe}\)}=\GF_{ij}\(\(x',\vec{\xe_{-i}}\);\vec{\xe}\),\]  
for all $j\neq i \in N$ and for all $x'\neq \xe$. This proves the claim.
\end{proof} 

Thus, \Cref{lem:equivalent} suggests that an allocation is evolutionary stable if and only if it is individually non-griefable. According to \Cref{def:griefable}, this is weaker than an allocation being non-griefable, which is satisfied if for all $i\in N$, the sum over $j\neq i \in N$ of all individual griefing factors $G_{ij}$ is less than $1$. 

\subsection{Griefing in Mining Games}\label{sec:part_1}

While immediate, \Cref{lem:equivalent} provides a handy way to generalize the notion of evolutionary stability. In particular, in general, non-homogeneous populations, we may impose the stability requirement that an allocation be individually non-griefable or, as mentioned above, the stronger requirement that an allocation be non-griefable. This brings us to the main result of this section, which suggests that the Nash equilibrium of \Cref{thm:arnosti} is griefable for both symmetric and asymmetric populations of miners. In particular, assuming that the network has stabilized at the $x^*$ equilibrium allocation, a strategic miner may attack other miners simply by increasing their own mining resources. Specifically, if a miner $i$ deviates to a resource allocation $x_i^*+\Delta$ for some $\Delta>0$, then this creates a GF equal to $\mathcal{O}\(n/\Delta\)$. Such a deviation reduces the attacking miner's own payoff but, as we will see, it decreases the payoff of all other miners by a larger margin. This improves the attacking miner's \emph{relative payoff} and hence their long-term survival chances in the blockchain mining network. This is formalized in \Cref{thm:grief}. All proofs of \Cref{sec:single} are presented in \Cref{app:omitted}.

\begin{theorem}\label{thm:grief}
Let $\Gamma=\gam$ be a mining game and let $\vec{x}^*=\(x^*_i\)_{i\in N}$ be its unique pure strategy Nash equilibrium.
\begin{enumerate}[leftmargin=0.6cm, label=(\roman*)]
\item In a homogeneous population, i.e., when all miners have the same cost, $c_i=c>0$ for all $i\in N$, the unique Nash equilibrium allocation $x^*=\frac{n-1}{n^2c}$ is not evolutionary stable. In particular, there exists $x'\neq x^*$, so that an individually deviating miner $i$ increases their relative payoff $\Pi_i\(x',x^*_{-i}\)-\Pi_j\(x',x^*_{-i}\)$.
\item In a general, non-homogeneous population, the pure Nash equilibrium $x^*$ is griefable. In particular, assuming that all miners $j\in N$ are using their equilibrium allocations $x_j^*, j\in N$, the deviation $x_i^*+\Delta$, for some $\Delta>0$, of miner $i\in N$, has a griefing factor
\[\GF_i\(\(x_i^*+\Delta,\vec{x}_{-i}^*\);\vec{x}^*\)=  \frac{n-1}{\Delta\cdot\sum_{j=1}^nc_j} = \mathcal{O}\(n/\Delta\).\]
In particular, at the Nash equilibrium allocation, $x^*$, any single miner may increase their mining resources and improve their utility in relative terms.
\item In both the homogeneous and non-homogeneous populations, the unique individually non-griefable allocation, $\vec{y}=\(y_i\)_{i\in N}$, satisfies $y_i=\frac{n}{n-1}x_i^*$, where $x^*_i$ is the Nash equilibrium allocation of miner $i\in N$.
\end{enumerate}
\end{theorem}
\begin{remark} Part (ii) of \Cref{thm:grief} reveals one shortcoming of the current definition of GF. Specifically, the GF may grow arbitrarily large as $\Delta\to 0$. However, as $\Delta\to0$, the \emph{absolute} total harm to the network is negligible (even if the relative loss is very large as expressed by the GF). One possibility to circumvent this problem is to consider discrete increments for $\Delta$, i.e., $\Delta \in \{1,2,\dots,100,\dots\}$ as in e.g., \cite{Chen19}. Alternatively, one may combine GF with the absolute loss of the network to obtain a more reliable measure. We do not go deeper into this question at the current moment since it seems to be better suited for a standalone discussion. We leave this analysis as an intriguing direction for future work.\par
\end{remark}
\begin{remark}\label{rem:part2} Part (iii) of \Cref{thm:grief} allows us to reason about the overall expenditure at the unique individually non-griefable allocation $\vec{y}=\(y_i\)_{i\in N}$. In the general case, that of a non-homogeneous population, the total expenditure at an individually non-griefable allocation $\vec{y}=\(y_i\)_{i\in N}$ is 
\begin{align*}
E\(\vec{y}\)=\sum_{i\in N}c_iy_i=\frac{n-1}{n}\sum_{i\in N}c_ix_i^*=n\lt 1-(n-1)\frac{\sum_{i}c_i^2}{\(\sum_{i}c_i\)^2}\rt,
\end{align*}
where we used that $x_i^*=(1-c_i/c^*)/c^*$ and $y_i=\frac{n}{n-1}x^*$ by \Cref{thm:arnosti} and part (iii) of \Cref{thm:grief}, respectively. Cauchy-Schwarz inequality implies that $\(\sum_{i}c_i\)^2\le n\sum_{i}c_i^2$ which yields that $E(y)\le 1$ with equality if and only if $c_i=c$ for all $i\in N$. Thus, the expenditure in the individually non-griefable allocation is always less than or equal to the aggregate revenue generated by mining, with equality only if the population is homogeneous. In that case, i.e., if all miners have the same cost $c_i=c$ for all $i \in N$, then the unique individually non-griefable allocation is also evolutionary stable (cf. \Cref{lem:equivalent}), i.e., $\vec{y}=\vec{\xe}$ with $\xe=\frac{1}{nc}$ for all $i\in N$ (by part (iii) and symmetry). In all cases, the total expenditure $E\(\vec{x}^*\)$, at the unique Nash equilibrium $\vec{x}^*$ must be equal to $E\(\vec{x}^*\)=\frac{n-1}{n}E\(\vec{y}\)$ and hence it less than the expenditure at the unique individually non-griefable allocation and strictly less than the generated revenue (which is equal to $1$). 
\end{remark}
%
In the proof of \Cref{thm:grief}, we have actually shown something slightly stronger. Namely, miner $i$'s individual loss due to its own deviation to $x_i^*+\Delta$ is less than the loss of each other miner $j$ provided that $\Delta$ is not too large. In other words, the individual griefing factors with respect to the Nash equilibrium allocation are all larger than 1 and hence, the Nash equilibrium is also individually griefable. This is formalized next.
\begin{corollary}\label{cor:single}
For every miner $j\in N$ such that $\Delta<x_j^*$, it holds that $\GF_{ij}\(\(x^*_i+\Delta,\vec{x^*_{-i}}\);\vec{x^*}\)>1$, i.e., the loss of miner $j$ is larger than the individual loss of miner $i$.
\end{corollary}

\Cref{thm:grief} and \Cref{cor:single} imply that miners are incentivised to exert higher efforts than the Nash equilibrium predictions. The effect of this strategy is twofold: it increases their own relative market share (hence, their long-term payoffs) and harms other miners. The notable feature of this \emph{over-mining} attack (or deviation from equilibrium) is that it does not undermine the protocol functionality directly. As miners increase their \emph{constructive effort} to, security of the blockchain network also increases. This differentiates the blockchain paradigm from conventional contests in which griefing occurs via exclusively destructive effort or deliberate sabotage against others \cite{Kon00,Ame12}.\par However, the over-mining strategy has implicit undesirable effects. As we show next, it leads to consolidation of power by rendering mining unprofitable for miners who would otherwise remain active at the Nash equilibrium and by raising entry barriers for prospective miners. This undermines the (intended) decentralized nature of the blockchain networks and creates long-term risks for its sustainability as a distributed economy. Again, this is a distinctive feature of decentralized, blockchain-based economies: for the security of the blockchain to increase, it is necessary that the aggregate resources \emph{and} their distribution among miners both increase (which is not the case in the over-mining scenario).

\begin{proposition}\label{prop:breakeven}
Let $\Gamma=\gam$ be a mining game with unique Nash equilibrium allocation $\vec{x^*}=\(x_i^*\)_{i\in N}$. Assume that all miners $j\neq i\in N$ are allocating their equilibrium resources $x_j^*$, and that miner $i$ allocates $x_i^*+\Delta$ resources for some $\Delta>0$. Then
\begin{enumerate}[leftmargin=0.5cm,label=(\roman*)]
\item the maximum increase $\Delta_i$ of miner $i$ before miner $i$'s payoff becomes zero is $\Delta_i=\frac{1}{c_i}-\frac{1}{c^*}$. 
\item the absolute losses of all other miners $j\neq i$ are maximized when $\Delta=\Delta_i$ and are equal to $c_ix_i^*$.
\end{enumerate}
\end{proposition}

\Cref{prop:breakeven} quantifies (i) the maximum possible increase, $\Delta_i$, in the mining resources of a single miner before their profits hit the break-even point (i.e., become zero), and (ii) the absolute losses of all other miners when miner $i$ increases their resources by some $\Delta$ up to $\Delta_i$. As intuitively expected, more efficient miners can cause more harm to the network (part (i)) and in absolute terms, this loss can be up to the equilibrium spending $c_ix_i^*$ of miner $i$, assuming that miner $i$ does not mine at a loss (part (ii)). While not surprising these findings provide a formal argument that cost asymmetries can be severely punished by more efficient miners and that efficient miners can grow in size leading ultimately to a centralized mining network.

\section{From Oligopoly to Market Equilibria}\label{sec:multiple}

The previous analysis hinges on an important assumption: namely, that each individual miner has a significant effect on aggregate market outcomes. The utility function in equation \eqref{eq:utility}
\[\Pi_i\(x_i,\vec{x}_{-i}\)=\frac{x_i}{X}v-c_ix_i, \qquad \text{for all } i\in N,\]
assumes that the allocation, $x_i$, of miner $i$ affects the aggregate market resources, since $X=x_i+X_{-i}$ (in the denominator of the proportional rewards of each miner $i\in N$). However, in large networks, individual resources are typically (or ideally) negligible in comparison to aggregate resources. With this in mind, the results that we derive with these utility functions can be interpreted as the existence of a positive feedback loop towards centralization: if miners are relative large to the size of the whole economy then, there are intrinsic motives for miners to cause griefing to their peers which leads to further concentration of resources in few miners.\par
This observation bring us to the next part of our analysis which concerns the study of same problem under the assumption that each miner has an individually insignificant influence on aggregate market outcomes. To study this setting in full generality, i.e., in the presence of multiple co-existing blockchain networks in which the miners may distribute their resources, we first introduce some additional notation.

\subsection{Additional Notation: Large Market Assumption and Quasi-CES Utilities}
As in \Cref{sec:model}, let $N=\{1,2,\dots,n\}$ denote the set of miners. In addition, let $M=\{1,2,\dots,m\}$ denote a set of $m$ mineable cryptocurrencies. Here the word \emph{mineable} refers to various possible mechanisms, such as Proof of Work, Proof of Stake or any other mining mechanism that requires proof (expense) of scarce resources. Let $c_{ik}, i\in N, k\in M$ denote the cost of miner $i$ to allocate one unit of resource in cryptocurrency $k$. Finally, let $v_k$ denote the aggregate revenue generated by cryptocurrency $k\in M$. Typically, $v_k$ refers to the newly minted coins and total transaction fees paid to miners within the study period. In the case of multiple blockchains, miner $i$'s utility in equation \eqref{eq:utility} can be generalized in a straightforward way to the following \emph{quasi-linear} utility
\begin{align}\label{eq:multiplex}
\Pi_{i}\(\vec{x}_i,\vec{x}_{-i}\)=\sum_{k=1}^m\frac{x_{ik}}{x_{ik}+\sum_{j\neq i}x_{jk}}v_{k}-\sum_{k=1}^mc_{ik}x_{ik}.
\end{align}
We will write $X_k:=\sum_{i\in N}x_{ik}$ to denote the aggregate allocated resources in blockchain $k\in M$. To make comparisons among different cryptocurrencies, it will be convenient to express all allocations in common monetary units that denote \emph{spending} rather than individual (and potential different) physical resources. Accordingly, let $b_{ik}:=c_{ik}x_{ik}$ denote the \emph{spending} of miner $i$ for cryptocurrency $k\in M$. A strategy of miner $i$ will be described by a non-negative vector $\vec{b}_i=\(b_{ik}\)_{k\in M}$. Using this notation, we can write equation \eqref{eq:multiplex} as 
\begin{align}\label{eq:linear}
\Pi_{i}\(\vec{b}_i,\vec{b}_{-i}\)&= \sum_{k=1}^m\frac{v_{k}}{c_{ik}X_k}\cdot c_{ik}x_{ik}-\sum_{k=1}^mc_{ik}x_{ik}= \sum_{k=1}^m v_{ik}b_{ik}-\sum_{k=1}^mb_{ik},
\end{align}
where $v_{ik}:=v_k/X_kc_{ik}$ for any $i\in N$ and $k\in M$. Equivalently, if $\bar{c}_k$ is such that $b_k:=\bar{c}_kX_k$ is the total spending of the network of cryptocurrency $k$, then $v_{ik}:=\(v_k/b_k\)\cdot\(\bar{c}_k/c_{ik}\)$ for any $i\in N, k\in M$. \par
The utility function in equation \eqref{eq:linear} assumes that miners are risk-neutral and that resources can be reallocated effectively in all networks. However, in practice this is not always the case. First, mining is largely an act of investment and as such it is subject to (considerable) risk. Each individual cryptocurrency market is subject to both volatile returns (fluctuations in the $v_{ik}$'s) and uncertainty concerning its future development and success. Thus, it is reasonable for individual miners to hedge their risks by diversifying their resources. Second, mining of a specific cryptocurrency typically requires a commitment in the invested resources (in form of mining equipment or staked capital). While in some cases, mobility of these resources can be assumed to be frictionless between different blockchains (e.g., when they use the same mining algorithm and technology), in general, this is not always the case.\par
To address these considerations, we introduce (as is standard in economics) \emph{diminishing marginal returns} from the mining revenues of each individual coin. This is captured via concave utility functions of the form $u_i\(x\)=x^{\rho_i}$, for some $0<\rho_i\le 1$, for each miner $i \in N$ which when aggregated, amount to a \emph{quasi Constant Elasticity of Substitution (quasi-CES)} utility function. Using this abstraction, miner $i$'s utility of equation \eqref{eq:linear} becomes
\begin{equation}\label{eq:ces}
\Pi_{i}\(\vec{b}_i,\vec{b}_{-i}\)=\(\sum_{k=1}^m\(v_{ik}b_{ik}\)^{\rho_i}\)^{1/\rho_i}-\sum_{k=1}^mb_{ik},
\end{equation}
Note that for $\rho_i=1$, we recover the quasi-linear utility of equation \eqref{eq:linear}. The parameters $\rho_i$ can be interpreted both as the risk profile of the miner and the mobility of their resources (depending on whether we view it as utility from consumption or utility from production). For instance, for $\rho_i\to0$, the utility \eqref{eq:ces} becomes a Cobb-Douglas utility which corresponds to maximum risk diversification (or equivalently minimal mobility of resources). In the other extreme, $q=1$ implies that the miner is risk neutral and can freely move their resources to the most profitable (in some correct sense) cryptocurrency. Intermediate values $0<\rho_i<1$ yield intermediate risk profiles and degrees of mobility of resources between different blockchains. We further discuss this topic in our case study in \Cref{sec:case}. \par
Finally, we assume that each miner $i\in N$ has a total monetary capacity, $K_i>0$, of resources and make the following important assumption. If the total capacity, $K_i$, of each individual miner $i\in N$ is not very large compared to the total allocated resources in each cryptocurrency, $X_k:=\sum_{i=1}^nx_{ik}$, in each cryptocurrency, then miner $i$ may neglect the effect of her own allocation, $x_{ik}$, in the total mining resources. In other words, each miner $i\in N$ takes the total mining capacity, $X_k$ of each cryptocurrency $k\in M$, as given in her strategic decision making. This implies that $v_{ik}$ does not depend on the decision of miner $i$ (nor on the decision of any other miner $j\in N$) and hence, the utility function in equation \eqref{eq:ces} is only a function of the $b_i$'s. We will denote the \emph{blockchain mining economy} defined by the utilities in equation \eqref{eq:ces} with $\Gamma=\(N,M,\(\Pi_i,v_{ik},\rho_i,K_i\)_{i\in N}\)$.

\subsection{Proportional Response Dynamics and Equilibrium Allocations} \label{sub:equilibrium}
The assumption that each individual miner has negligible influence in aggregate market outcomes has far-reaching implications in the equilibrium analysis of the blockchain mining economy $\Gamma$. Under this assumption, $\Gamma$ can be abstractly seen as a \emph{Fisher market with quasi-CES utilities}. This provides an alternative approach to determine its equilibria via the convex optimization tools that have been developed for the analysis of such markets \cite{Dev09,Bir11,Col17,Che18}. \par
Based on this framework, we derive a \emph{Proportional Response (PR)} update rule that converges to the equilibrium of this economy for any selection of the $\rho_i's \in(0,1]$. Since the utilities in our case are quasi-CES, we need to adapt existing techniques (which are available only for linear or quasi-linear cases). This is topic of this section which leads to the main result that is stated in \Cref{thm:equilibrium}. To focus on the interpretation of the results in the blockchain context rather than on the techniques, we defer all proofs to \Cref{app:proportional}. However, we note that the convergence result of the PR dynamics applies to \emph{any} Fisher markets with quasi-CES utilities and may be thus, of independent interest.\par
To formulate the proportional response dynamics, we first introduce some minimal additional notation. At time step $t\ge0$, let $u_{ik}\(t\):=\(v_{ik}b_{ik}\(t\)\)^{\rho_i}$ and $u_i\(t\):=\sum_{k=1}^m u_{ik}\(t\)$ denote miner $i$'s utility from cryptocurrency $k$ and aggregate utility (before accounting for expenses), respectively. Let also $w_i\(t\):=K_i-\sum_{k=1}^mb_{ik}\(t\)$ denote miner $i$'s unspent budget at time $t\ge0$ and let $\tilde{K}_i\(t\):=K_i\cdot \(K_i-w_i\(t\)\)^{\rho_i-1}$ for each $i\in N$. Then, for the utility function in \eqref{eq:ces}, we define the \emph{Proportional Response (PR) Dynamics} as follows
\begin{equation}\label{eq:proportional}
b_{ik}\(t+1\):=K_i\cdot\dfrac{u_{ik}\(t\)}{\max{\{u_i\(t\),\tilde{K}_i}\(t\)\}} \tag{PR}.
\end{equation}

A pseudocode implementation for the \eqref{eq:proportional} dynamics is provided in \Cref{alg:proportional}.

\begin{algorithm}[!htb]
\caption{\textsf{PR-QCES} Protocol}\label{alg:proportional}
\vspace*{0.1cm}
\raggedright
\textbf{Input (network):} network hashrate, $X_k$, and revenue, $v_k$, of each cryptocurrency $k\in M$.\\
\textbf{Input (miner):} miner $i$'s unit cost, $c_{ik}$, budget capacity, $K_i$, and utility parameter, $\rho_i$.\\
\textbf{Output:} equilibrium spending (allocation) $b_{ik}, k \in M$ for each miner $i\in N$.\\[-0.3cm]
\begin{algorithmic}[1]
\Initialize {spending (allocation) $b_{ik}>0$ for all $k\in M$.}
\Loop{over $t\ge0 $ till convergence}
\For {each miner $i\in N$}\vspace*{0.1cm}
\EndFor
\Procedure{Auxiliary}{$\(X_k,v_k,c_{ik}\)_{k\in M},K_i,\rho_i$}
\State{$v_{ik}\gets v_k/X_kc_{ik}$}
\State{$w_i \gets K_i-\sum_{k\in M} b_{ik}$}$\Comment{\text{not invested capital}}$
\State{$\tilde{K}_i\gets K_i\(K_i-w_i\)^{\rho_i-1}$} 
\State{$u_{ik}\gets \(v_{ik}b_{ik}\)^{\rho_i}$ and $u_i\gets  \sum_{k\in M} u_{ik}$} $\Comment{\text{utilities before subtracting costs}}$
\EndProcedure\vspace*{0.1cm}
\Procedure{PR-Dynamics}{$\(u_{ik}\)_{k\in M},u_i,K_i,\tilde{K}_i$}
\If {$u_i>\tilde{K}_i$} 
\State $b_{ik}\gets u_{ik}K_i/u_i$
\Else 
\State $b_{ik}\gets u_{ik}K_i/\tilde{K}_{i}$
\EndIf
\EndProcedure
\vspace*{0.1cm}
\State{$X_k\gets \sum_{j\in N}b_{jk}$} $\Comment{\text{update network hashrate and repeat}}$
\EndLoop
\end{algorithmic}
\end{algorithm} 

An important feature of the PR update rule is that it has low informational requirements. It uses as inputs only observable information at network level (aggregate revenue and hashrate) and local information at a miner's level (individual capacity and mining cost). Thus, it provides a protocol that is both feasible to implement in practice and relevant for this particular type of large, distributed economics.\par
Intuitively, the update rule \textsf{PR} suggests the following. If the revenue of miner $i$ is high enough at round $t$, i.e., if $u_i\ge\tilde{K}_i$, then miner $i$ will reallocate all their resources in round $t+1$ in proportion to the generated revenues, $u_{ik}/u_i$, in round $t$. By contrast, if $u_i< \tilde{K}_i$, then miner $i$ will behave cautiously and allocate only a fraction of their resources. This fraction is precisely equal to the generated revenue at round $t$, i.e., $u_i\(t\)$ again in proportion to the revenue generated by each cryptocurrency $k\in M$. The important property of the PR dynamics is that they converge to the set of equilibrium allocations for any initial strictly positive allocation vector. This is statement of \Cref{thm:equilibrium} which is our main theoretical result. 
 
\begin{theorem}[Mining Resources Equilibrium Allocation]\label{thm:equilibrium}
For any positive initial allocation, $\vec{b}^0>0$, the \eqref{eq:proportional}-dynamics converge to the set of equilibrium allocations, $\vec{b}^*$, of the blockchain mining economy $\Gamma$.
\end{theorem}

\Cref{thm:equilibrium} will be our main tool to study equilibria in the blockchain mining economy. Before we proceed with our empirical results in \Cref{sec:case}, a comparison between the oligopoly model of \Cref{sec:part_1} and the market model of \Cref{sec:multiple} is due.

\subsection{Comparing the two Models: Nash vs Market equilibria}\label{sub:comparison}
When comparing the two models that we considered thus far, the oligopoly model in \Cref{sec:part_1} and the market model in \Cref{sec:multiple}, we make the following two main observations.

\paragraph{Bounded versus Unbounded Capacities} The first observation concerns the capacities of individual miners and the influence that they have on aggregate outcomes. In the former case, that of the oligopoly model, miners are assumed to have large capacities of resources which if used strategically, affect the welfare of other miners. This generates adversarial incentives that lead to griefing and destabilize the Nash equilibrium outcome. At the unique evolutionary stable equilibrium, griefing is not possible, however, at that equilibrium, miners fully dissipate the reward and cause further consolidation of power that raises entry barriers to prospective entrants. Thus, the system enters a positive feedback loop towards market concentration. \par
By contrast, in the latter case, that of the market model, miners are assumed to have individually negligible resources in comparison to aggregate market levels. Moreover, the comparison of instances with few large miners to instances with many small miners, necessitates the introduction of capacity limits to the model. Thus, miners cannot arbitrarily increase their allocated resources to harm others (and benefit themselves in relative terms). This renders griefing irrelevant and is the decisive factor that ultimately leads to stabilization (see \cite{Puu08} for a related argument). One may argue that the oligopoly model is closer to what we observe in practice, whereas the market model is the ideal model that was envisioned in the paper by Nakamoto that sparked the interest in the blockchain-based economies \cite{Nak08}. 

\paragraph{Equilibrium versus Learning Dynamics}
The second observation concerns the discrepancy between the static approach in the oligopoly model and the dynamic approach in the market model. The question that naturally arises is whether typical learning dynamics converge to the equilibrium of the oligopoly model. The answer to this question is negative and highlights the necessity of the market assumption (individually negligible resource capacities) to obtain a stable update rule that converges to equilibrium. \par
To see this, we consider two greedy update rules, \emph{Gradient Ascent} and \emph{Best Response} dynamics that are frequently used in such strategic interactions. Recall that the utility of miner $i$ in the strategic (single blockchain) model is given by
\[\Pi_{i}\(x_i,X_{-i}\)=\frac{x_i}{x_1+X_{-i}}-c_ix_i, \; \text{for all } x_i\ge0\]
and $i=1,2,\dots,n,$ where $X_{-i}=\sum_{j\neq i}x_j$ (cf. equation \eqref{eq:utility}). Thus, the \emph{Gradient Ascent (GA)} update rule is given by 
\begin{equation}\label{eq:ga}
x_i^{t+1}=x_i^t+\theta_i\frac{\partial}{\partial x_i^t}\Pi_i\(x_i^t\) = x_i^t+\theta_i\left[\frac{X^t_{-i}}{\(x_1^t+X_{-i}^t\)^2}-c_i\right], \tag{GA}
\end{equation}
for all $i=1,2,\dots,n$, where $\theta_i$ is the learning rate of miner $i=1,2$. The bifurcation diagrams in \Cref{fig:bifga} show the attractor of the dynamics for different values of the step-size (assumed here to be equal for all miners for expositional purposes) and for different numbers of active miners, $n=2,5$ and $10$, with $c_i=1$ for all $i$. The blue dots show the aggregate allocated resources for 400 iterations after a burn-in period of 50 iterations (to ensure that the dynamics have reached the attractor). \par
All three plots indicate that the GA dynamics transition from convergence to chaos for relative small values of the step-size. Interestingly, as the number of miners increases, the instabilities emerge for increasingly smaller step-size. This is in sharp contrast to the (PR) dynamics and their convergence to equilibrium under the large market assumption. The reason that a growing number of miners does not convey stability to the system is precisely because the miners are not assumed to have binding capacities. As miners act greedily, their joint actions drive the system to extreme fluctuations and the larger their number, the easier it is for these fluctuations to emerge. Finally, while convergence is theoretically established for small step-sizes in all cases, such step-sizes correspond to very slow adaption and are of lesser practical relevance.

\begin{figure}[!htb]
\centering
\includegraphics[width=0.31\linewidth]{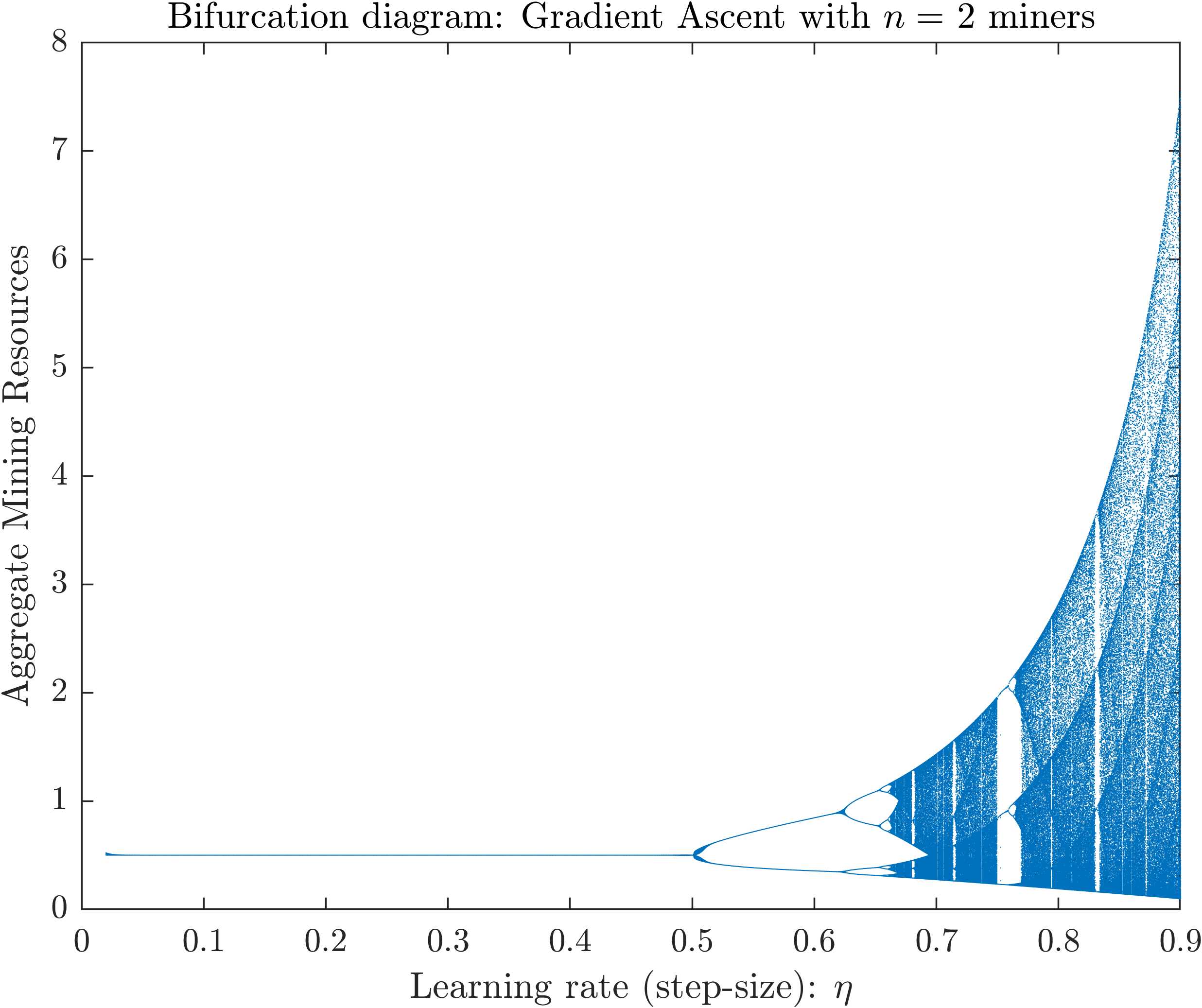}\hspace{1pt}
\includegraphics[width=0.31\linewidth]{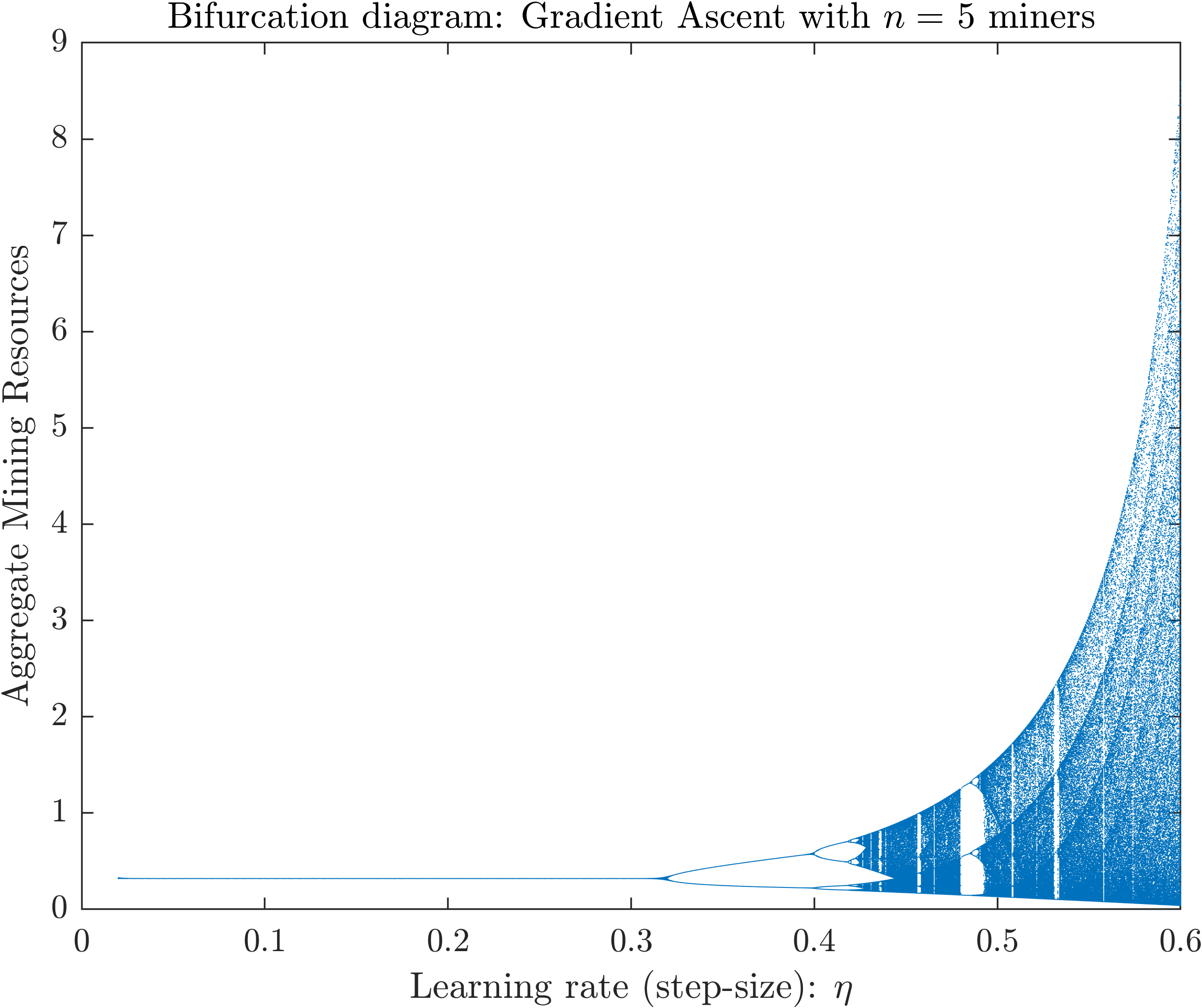}\hspace{1pt}
\includegraphics[width=0.31\linewidth]{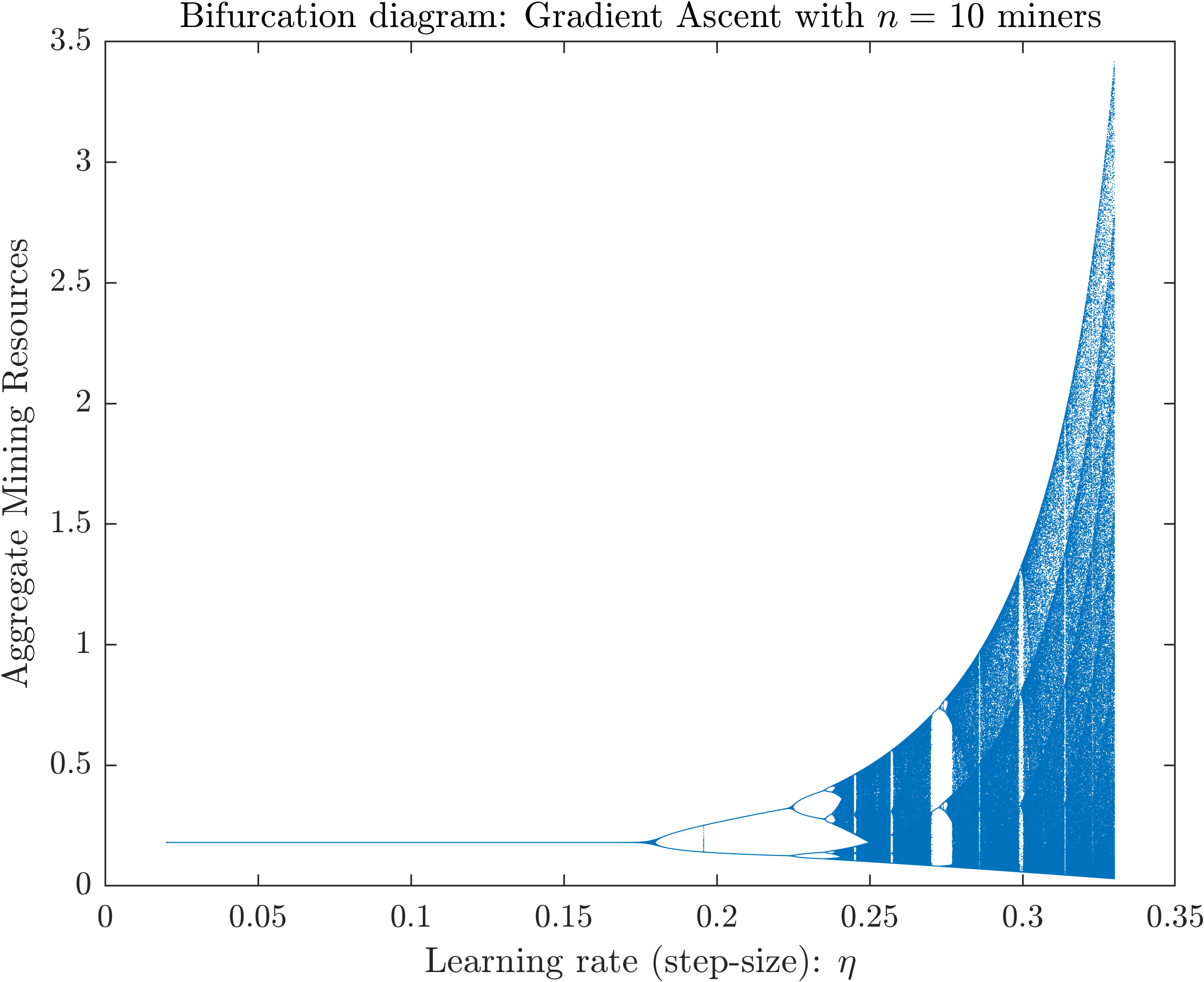}
\caption{Bifurcation diagrams for the Gradient Ascent dynamics with $n=2,5,10$ miners with respect to the learning parameter $\theta$. As the number of miners grows, the dynamics become chaotic for even lower step-sizes. }
\label{fig:bifga}
\end{figure}

We obtain a qualitatively similar result for the best response dynamics. The \emph{Best Response (BR)} update rule is given by 
\begin{equation}\label{eq:br}
x_i^{t+1}=\sqrt{X_{-i}^t/c_i}-X_{-i}^t,\quad \text{for all } i=1,2,\dots,n. \tag{BR}
\end{equation}
As above, the bifurcation diagrams in \Cref{fig:bifbr} show the aggregate allocated mining resources for $n=2,5$ and $10$ miners. The horizontal axis (i.e., the bifurcation parameter) is now the cost asymmetry between the representative miner and all other miners which are assumed to have the same cost (again only for expositional purposes). The plots suggest that the stability of the dynamics critically depend on the parameters of the system with chaos emerging for various configurations. 

\begin{figure}[!htb]
\centering
\includegraphics[width=0.31\linewidth]{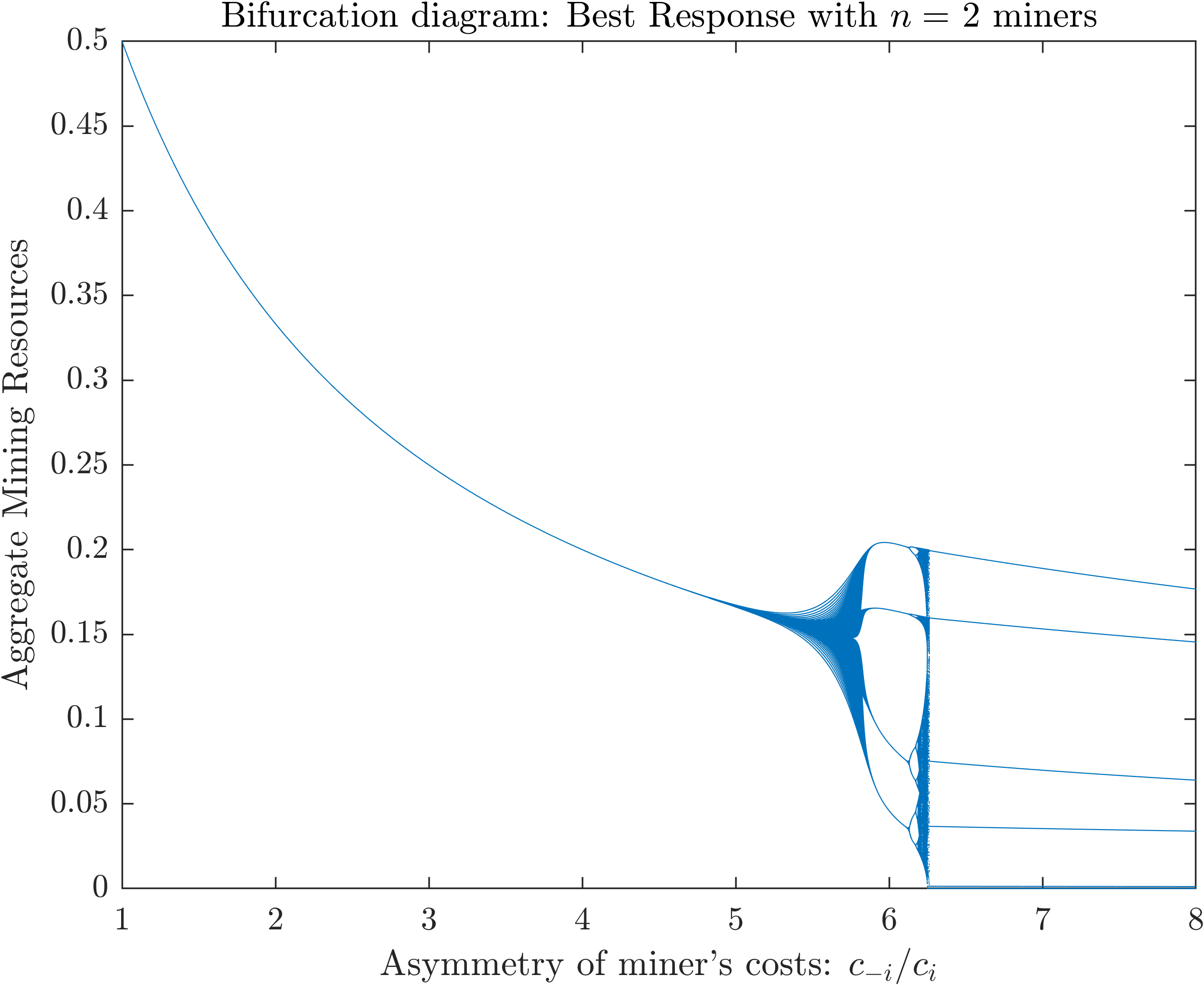}\hspace{1pt}
\includegraphics[width=0.31\linewidth]{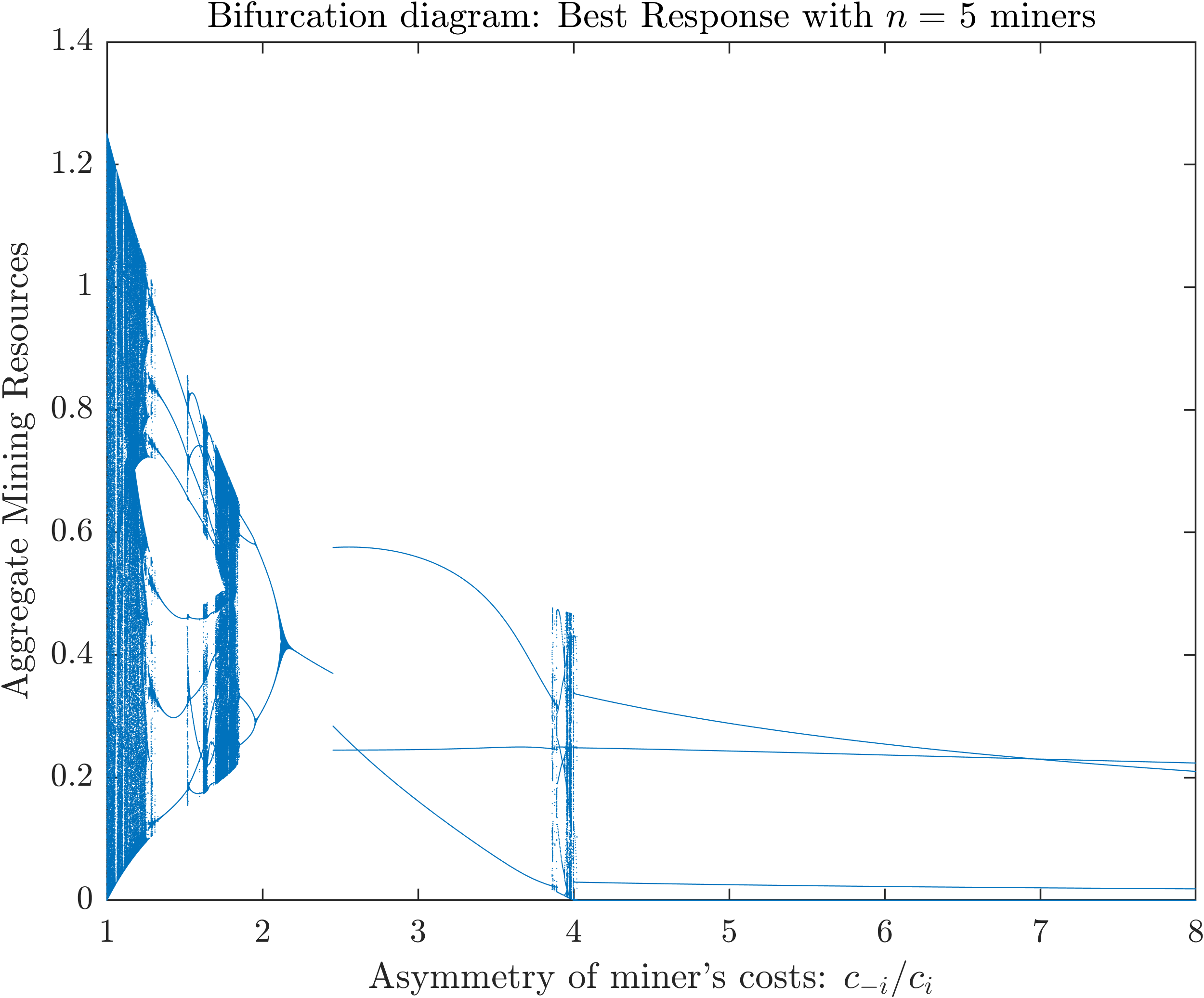}\hspace{1pt}
\includegraphics[width=0.31\linewidth]{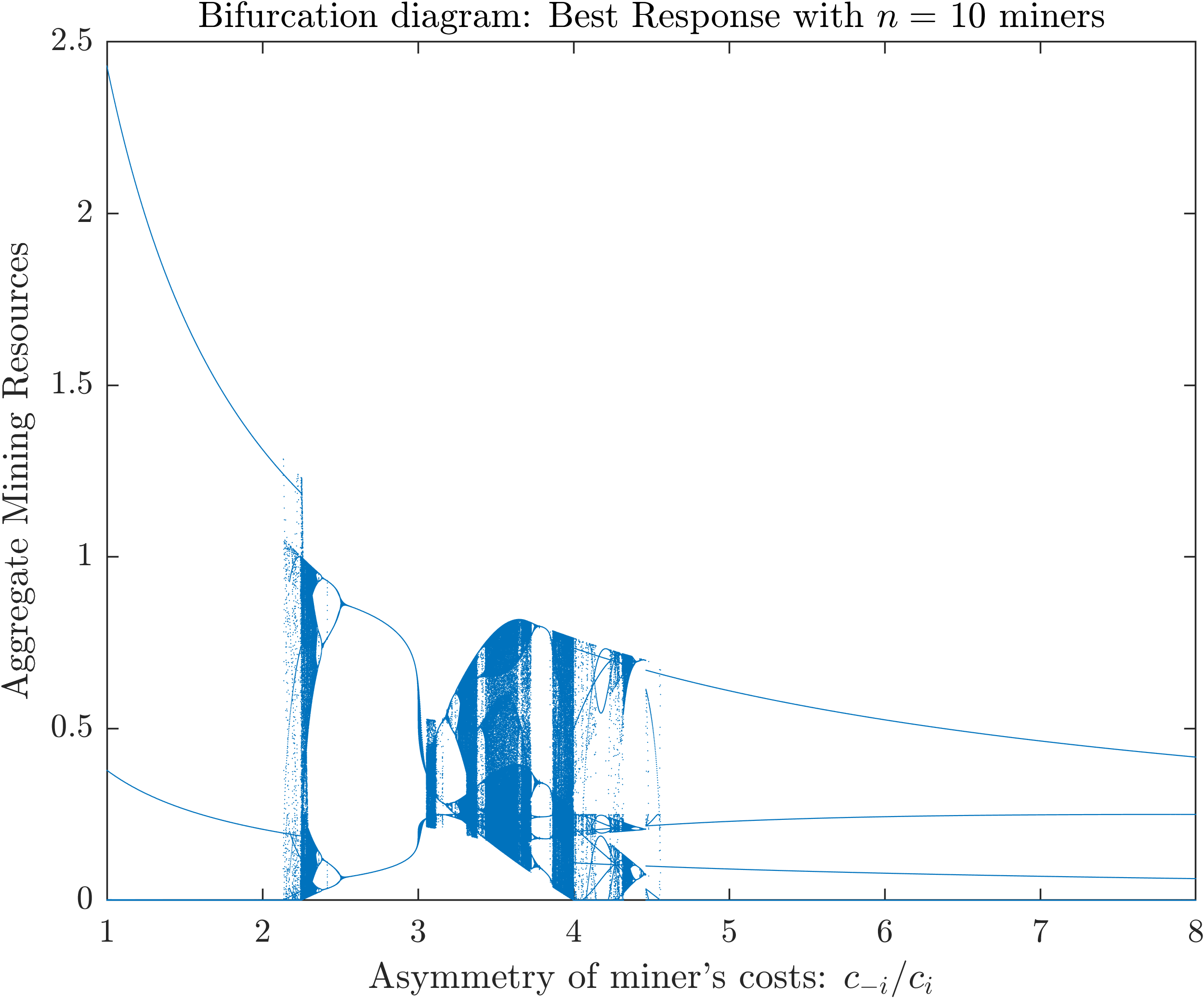}
\caption{Bifurcation diagrams for the Best Response dynamics with $n=2,5,10$ miners with respect to the miners cost asymmetry. The dynamics become chaotic typically for intermediate values of cost asymmetry. }
\label{fig:bifbr}
\end{figure}

In sum, the above results indicate the importance of the large market assumption, i.e., that miners' individual allocations do not affect aggregate network levels, in the stability of the blockchain ecosystem. As showcased by the GA and BR dynamics, if miners' decisions affect the decisions of other miners and if miners can adjust (increase or decrease) their capacities to optimize their profits, 
then common learning dynamics can exhibit arbitrary behavior. Instead of converging to the Nash equilibrium (or to some other stable outcome), the aggregate allocations may oscillate between extreme values or exhibit chaotic trajectories, with adverse effects on the reliability of the supported applications and the value of the blockchain-based cryptocurrency. Along with our earlier findings about griefing, these results paint a more complete picture about the various reasons that can destabilize permissionless blockchain networks when there is concentration of mining power.

\section{Case Study: Allocation of Mining Resources in the Wild}\label{sec:case}
Due to its low informational requirements, the PR protocol allows us to reduce the degrees of freedom that accompany synthetic data (such as estimates about the numbers of active miners, their individual capacities, mining costs etc.), and adopt a \qt{single miner's} perspective against real data when we study equilibrium allocations in the actual blockchain mining economy. Since the PR protocol converges to the equilibrium allocations for a wide range of quasi-CES utilities (as defined by parameters $\rhoi \in (0,1]$), we can reason about the effects of risk diversification and resource mobility on miners' equilibrium allocations. This allows us to extend existing results \cite{Jin20,Shu20}. 

\subsection{Data Set and Experimental Setting}\label{sub:data}
We apply the above theoretical framework in the following case study in which we consider four Proof of Work blockchains (cryptocurrencies): Bitcoin (BTC), Bitcoin Cash (BCH), Ethereum (ETH) and Litecoin (LTC). Our data set consists of the total daily network hashrate in TeraHashes per second (TH/s) and the aggregate daily miners' revenue in USD (newly minted coins and transaction fees) the for the four selected cryptocurrencies in the period between 1/1/2018 and 10/18/2020. The data are visualized in \Cref{fig:hashrates,fig:revenues}.

\begin{figure}[!htb]
\includegraphics[width=\linewidth]{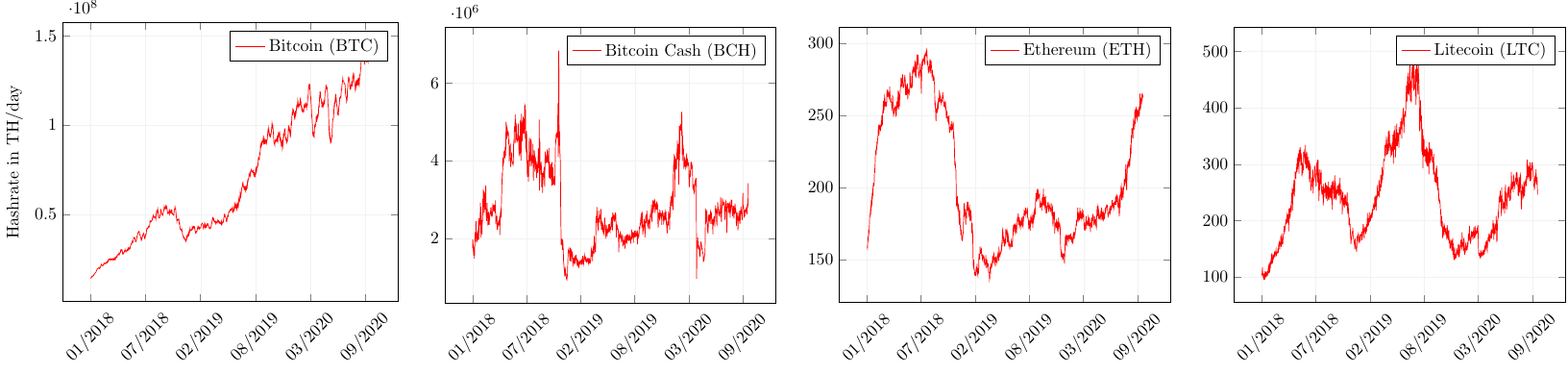}
\caption{Daily estimated hashrate (measured in TeraHashes per day (TH/day)) in the four cryptocurrencies: Bitcoin (BTC), Bitcoin Cash (BCH), Ethereum (ETH) and Litecoin (LTC). Source: \href{https://glassnode.com/}{glassnode.com}.}\label{fig:hashrates}
\end{figure}

\begin{figure}[!htb]
\includegraphics[width=\linewidth]{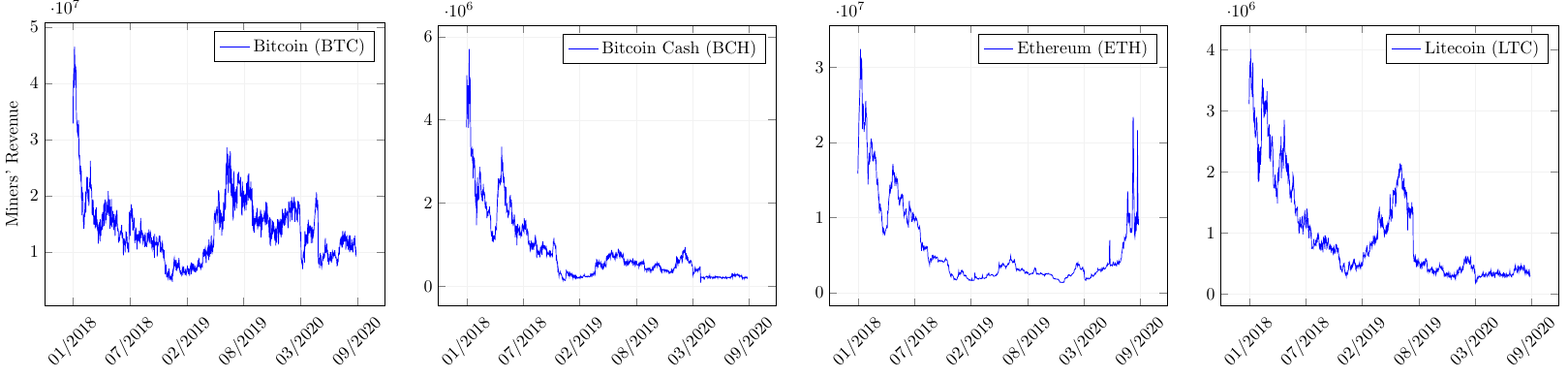}
\caption{Daily miners' revenue (aggregate value in USD of newly minted coins and transaction fees) in the four cryptocurrencies: Bitcoin (BTC), Bitcoin Cash (BCH), Ethereum (ETH) and Litecoin (LTC). Source: \href{https://glassnode.com/}{glassnode.com}.}\label{fig:revenues}
\end{figure}

To apply the PR-QCES protocol, we need to derive an estimation for the cost of a representative miner to produce one unit of resource, i.e., one TH/s for a whole day, in each network. This is done as follows. For each cryptocurrency, we collect data regarding the state of the art (or most popular) mining equipment for each calendar year in the considered time period. The data (and their sources) are presented in \Cref{tab:equipment}.

\begin{table*}[!htb]
\centering
\arrayrulecolor{blue!30}
\setlength{\tabcolsep}{8pt}
\begin{tabular}{@{}llllrr@{}}
\toprule
& \textbf{Year} & \textbf{Model} & \textbf{Price $(P)$} & \textbf{Hashrate $(H_s)$}& \textbf{Power $(W)$}\\
\midrule
\multirow{3}{51pt}{Bitcoin/ Bitcoin Cash}& 2018 & Ebang Ebit E11+ & \$2,494 & 37 TH/s & 2035W\\
&2019 & Antminer s17 & \$2,100 & 56 TH/s & 2520W\\
&2020 & Antminer s19 Pro & \$2,507 & 110 TH/s & 3250W\\
\midrule
\multirow{3}{50pt}{Ethereum}
&2018 & PandaMiner B5+ & \$2,916 & 110 MH/s & 800W\\
&2019 & PandaMiner B7 & \$2,035 & 230 MH/s & 1150W\\
&2020 & PandaMiner B9 & \$3,280 & 330 MH/s & 950W\\
\midrule
\multirow{3}{50pt}{Litecoin}
&2018 & Moonlander 2 L3++ & \$65 & 5 MH/s & 10W\\
&2019 & FutureBit Apollo LTC & \$500 & 120 MH/s & 200W\\
&2020 & Antminer s19 Pro & \$300 & 580 MH/s & 1200W\\
\bottomrule
\end{tabular}
\caption{Mining equipment. The selected models correspond to the state of the art or most popular mining rigs for each cryptocurrency. The lifespan, $L_s$, of all model is assumed to be 2 years. Sources: \href{https://www.asicminervalue.com/miners}{(asicminervalue.com)} for Bitcoin and Bitcoin Cash, \href{https://www.pandaminer.com}{(pandaminer.com)} for Ethereum, and \href{https://www.exodus.io/blog/litecoin-mining-hardware}{(exodus.io)} for Litecoin. }
\label{tab:equipment}
\end{table*}

The hardest part in the data collection process is the estimation of a single average value for the average network cost per kWh. According to \cite{Vri18,Vri20} prices per kWh follow a seasonal trend (due to weather dependent fluctuations, e.g., in China) and a constant to slightly decreasing overall trend between 2018 and 2020. The exact values that we used in the experiments are in \Cref{tab:kWh}. However, as argued by the referenced papers, these estimates should be accepted with caution.\par
Using the above figures, the cost, $c$, to produce one TH/s for a whole day is given by the following formula
\[c = \frac{P}{365\cdot L_s\cdot H_s}+\frac{(W/1000)\cdot c(kWh)\cdot 24}{H_s},\]
where, as in \Cref{tab:equipment,tab:kWh}, $P$ denotes the acquisition price of the model in USD, $L_s$ the useful lifespan (assumed to be 2 years for all models), $H_s$ the effective hashrate of the model (in TH/s), $W$ its power consumption (in Watt) and $c(kWh)$ the average cost per kWh in USD.
\begin{table*}[!htb]
\centering
\arrayrulecolor{blue!30}
\setlength{\tabcolsep}{13pt}
\begin{tabular}{@{}cccccc@{}}
\toprule
\multicolumn{6}{c}{\textbf{Average price per kWh} $(c(kWh))$}\\
\midrule
01-06/2018 & 07-12/2018 &01-06/2019 &07-12/2019 &01-06/2020 &07-12/2020 \\
\$0.06 & \$0.05 & \$0.04 & \$0.05 & \$0.03 & \$0.02\\
\bottomrule
\end{tabular}
\caption{Average prices per kWh. The figures are updated every six months (i.e., 01-06 and 07-12 in each year) and concern a global estimated average. They are mainly based on \cite{Vri18,Vri20}. Scattered online resources offer similar estimates but we refrain from recommending these figures as precise. As cautioned by the referenced papers, there are several practical reasons for which these figures may have limited accuracy, e.g., different bargains achieved by individual (large) miners, lack of transparency in the exact energy source (renewable or electricity), spatial and seasonal fluctuations in prices (even within the same country as in the USA or China) etc. In the context of the current empirical study, the exact trend and values of electricity prices do not affect the interpretation of the results.}
\label{tab:kWh}
\end{table*}

The above estimations provide the necessary inputs to run the PR-QCES protocol (cf. Inputs in \Cref{alg:proportional}) for a single miner with cost $c$ per kWh and obtain their equilibrium allocations given the network hashrates and aggregate revenue for each coin. Throughout, we assume that the miner has a fixed capacity, $K_i$, which is a small percentage of the aggregate mining resources in all networks. Different values of $K_i$ yield the same equilibrium allocations and are hence not presented here. In each experiment, we use a different parameter, $\rho_i$, in the quasi-CES utility of miner $i$. This allows us to reason about the effects of a miner's risk profile or under a different interpretation, of the degree of mobility of their resources, against real data. 

\subsection{Empirical Results}\label{sub:empirical}

Our results are summarized in \Cref{fig:abstract}. 
\begin{figure}[!htb]
\centering
\includegraphics[width=\linewidth]{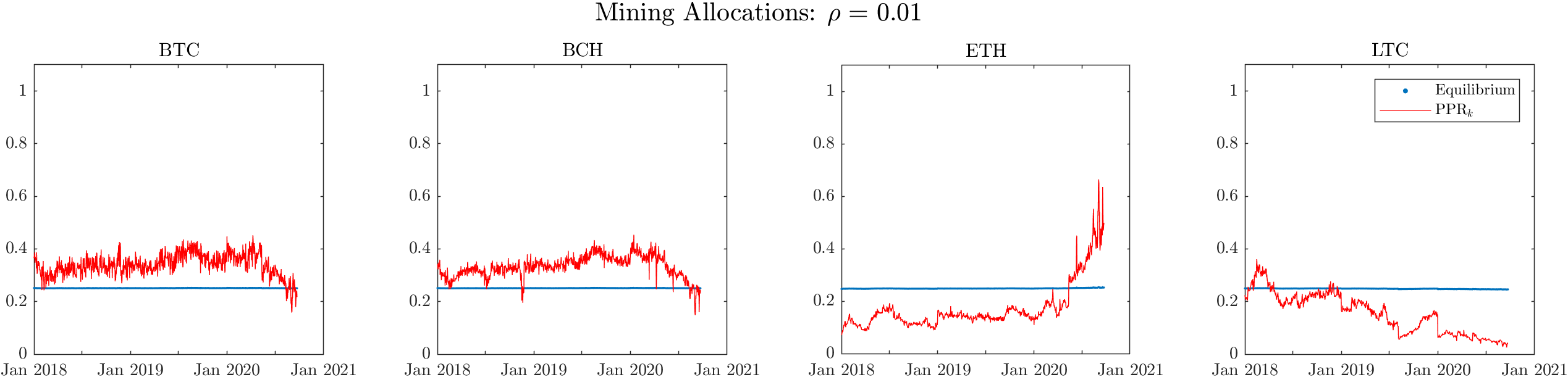}\\[0.3cm]
\includegraphics[width=\linewidth]{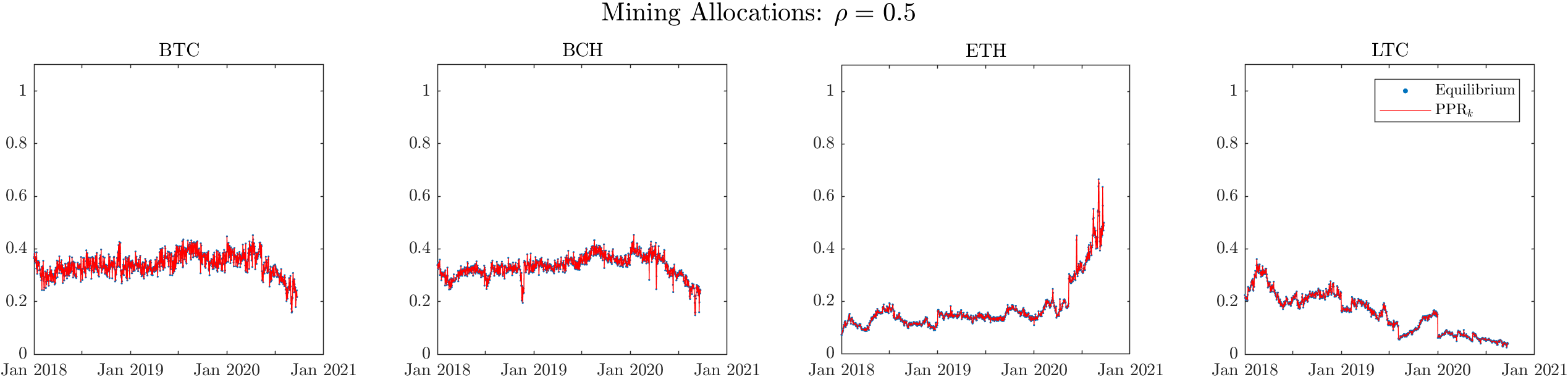}\\[0.3cm]
\includegraphics[width=\linewidth]{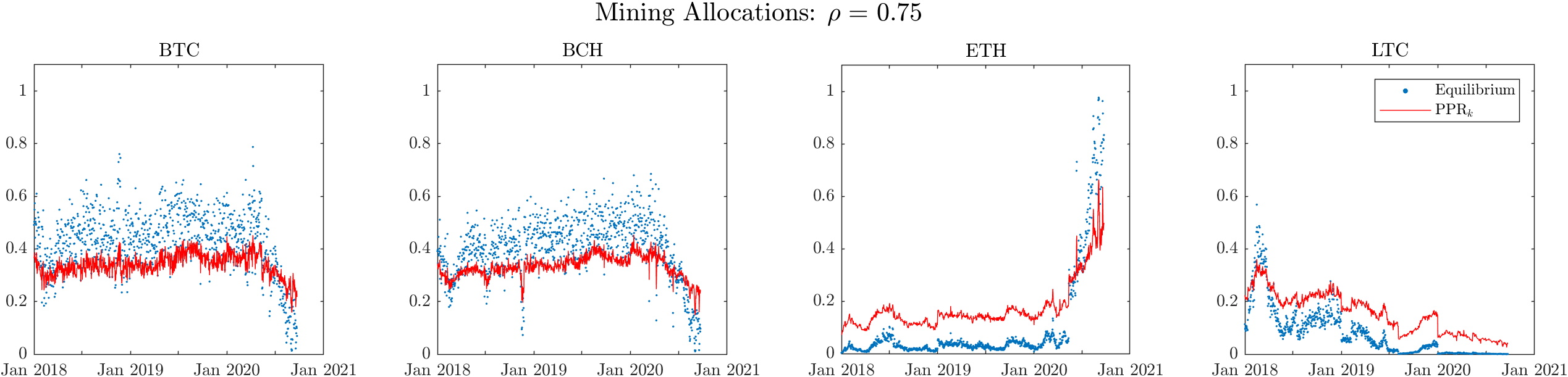}\\[0.3cm]
\includegraphics[width=\linewidth]{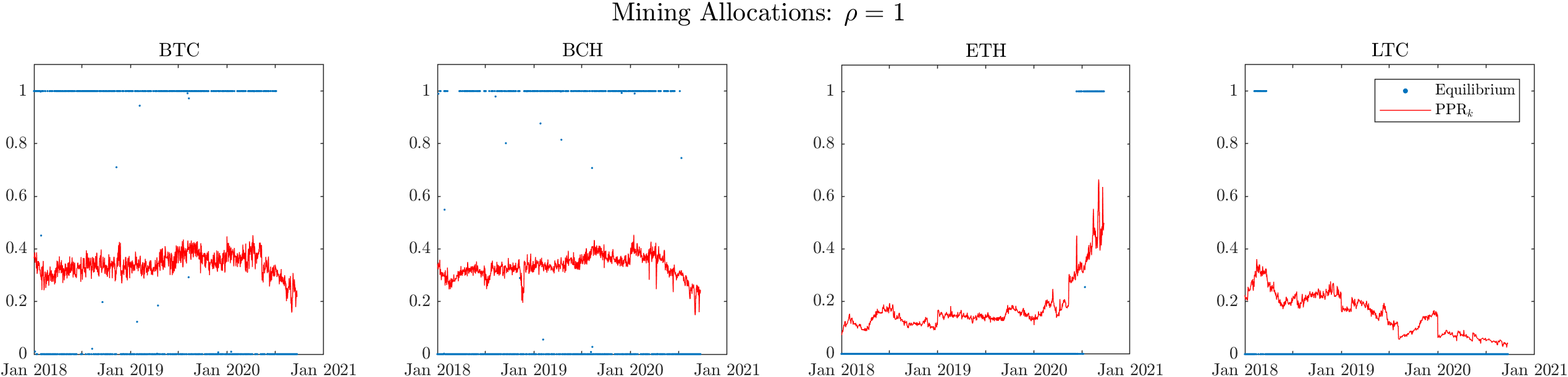}\\[0.3cm]
\caption{Equilibrium allocations (blue dotted lines) and proportional profitability ratios ($PPR_k$'s) (red lines) for the four coins of the case study. Each row corresponds to a different parameter $\rhoi$ of the miner's quasi-CES utility. When $\rhoi=0.01$, the miner distributes equally their resources among the available networks (risk aversion). For $\rhoi=0.5$, the equilibrium allocations exactly match the $PPR_k$'s (cf. \Cref{def:ppr}). For $\rho=1$, which corresponds to risk neutrality (with full mobility of resources), the miner allocates all their resources to the cryptocurrency with the highest $PPR_k$.}
\label{fig:abstract}
\end{figure}
Each row of \Cref{fig:abstract} corresponds to a different parameter $\rhoi$ and each panel corresponds to each of the four considered coins: Bitcoin (BTC), Bitcoin Cash (BCH), Ethereum (ETH) and Litecoin (LTC). In each panel, the blue dotted lines depict the equilibrium allocations of the miner for each day (derived by running the PR-QCES dynamics) for that coin and the red lines depict the \emph{proportional profitability ratio} of the coin which turns out to play an important role in the interpretation of the results. Formally, the proportional profitability ratio of a coin is defined as follows. 

\begin{definition}[Profitability and Proportional Profitability Ratios]\label{def:ppr}
Let $v_k$ denote the aggregate network revenue and $b_k$ the aggregate network spending (e.g., hashrate times cost to produce this hashrate) for mining cryptocurrency $k=1,2,\dots,n$ in a specific time period (e.g., one day). The \emph{profitability ratio}, $(PFR_k)$, of cryptocurrency $k$ is defined by
\begin{equation}\label{eq:pfrk}
PFR_k := \frac{\text{total network revenue from mining coin $k$}}{\text{total network spending for mining coin $k$}}=\frac{v_k}{b_k} \, .
\end{equation}
The \emph{proportional profitability ratio}, $(PPR_k)$, of cryptocurrency $k$ is defined as the ratio of $PFR_k$ over the sum of the $PFR_k$'s of all considered cryptocurrencies, i.e., 
\begin{equation}\label{eq:pprk}
PPR_k :=  \frac{PFR_k}{\sum_{j=1}^nPFR_j}=\frac{v_k/b_k}{\sum_{j\in M}v_j/b_j}\,. 
\end{equation} 
\end{definition}

We this definition at hand, we return to the interpretation of the results in \Cref{fig:abstract}. The first row shows the equilibrium allocations of a risk averse miner with $\rhoi=0.01$. Such a miner essentially ignores the input data and distributes (approximately) evenly their resources among the available coins (for $\rhoi=0$, the distribution would be exactly uniform, i.e., $1/4$ for each coin). The second row shows the equilibrium allocations of a miner with $\rhoi=0.5$. This value of parameter $\rhoi$ suggests that the miner is still willing to diversify their risks, albeit to a lesser extent. From a production perspective, $\rhoi=0.5$ implies an intermediate degree in the mobility of resources. Such a degree may be viewed as realistic in practical applications since there exist some blockchains that use compatible technology (e.g., Bitcoin and Bitcoin Cash) and certain mining models which are easily switchable between different mining algorithms. Interestingly, this case reveals the empirical finding that the a miner with parameter $\rhoi=0.5$ allocates their resources precisely according to the $PPR_k$ of each coin $k=1,2,3,4$ (cf. \Cref{def:ppr}). Thus, such a miner can fully determine their allocations from observable network data (aggregate revenues and hashrate) and local information (their own mining cost and capacity). \par
The importance of the $PPR_k$ is further highlighted in the last row of the matrix which shows the equilibrium allocations for a risk neutral miner with $\rho_i=1$. Such a miner allocates on each day (or period) the entirety of their resources to the cryptocurrency with the highest $PPR_k$. This approach is consistent with full or instant mobility of resources and can be, thus, observed in practice only between cryptocurrencies that use the same mining technology such as Bitcoin and Bitcoin Cash. From a modeling perspective, it highlights the importance of considering quasi-CES utilities instead of quasi-linear utilities in the case of multiple co-existing blockchains. From an analytical perspective, it also highlights how the use of different mining technologies by different blockchains 
conveys stability to the blockchain ecosystem as a whole by acting as a barrier in arbitrary reallocations of resources. Finally, the fourth row includes an intermediate case with $\rhoi=0.75$.

\section{Conclusions}\label{sec:conclusions}
In this paper, we studied resource allocation in blockchain mining networks. We identified two very different reasons for instabilities in the mining networks when mining power is consolidated in few miners: griefing (which generalizes the notion of evolutionary stability to non-homogeneous populations) and instability of dynamic allocation rules (such as gradient ascent or best response). Along with existing in-protocol attacks, such as selfish mining or manipulation of the difficulty adjustment in Proof of Work blockchains (\cite{Gar15,Eya18,But19} and \cite{Fia19, Gor19, Shu20}), these results paint a more complete picture of the inherent instabilities of these decentralized networks in practice. By contrast, under a large market assumption, which can be met in practice as more miners enter the blockchain ecosystem, we show that  these problems disappear and we establish convergence of a natural proportional response protocol to non-griefable market equilibria. The protocol has low informational requirements which make it suitable for such decentralized settings and converges to the market equilibria for a wide range of miners' risk diversification and various degrees of resource mobility between different blockchain networks. Our theoretical and empirical results suggest that decentralization, risk diversification among different blockchains and restricted mobility of resources (as enforced by the use of different mining technologies among different blockchains) are all factors that contribute to the stabilization of this otherwise volatile and unpredictable ecosystem. 

\section*{Acknowledgments}
This research is supported in part by NRF2019-NRF-ANR095 ALIAS grant, grant PIE-SGP-AI-2018-01, NRF 2018 Fellowship NRF-NRFF2018-07,  AME Programmatic Fund (Grant No. A20H6b0151) from the Agency for Science, Technology and Research (A*STAR) and the National Research Foundation, Singapore under its AI Singapore Program (AISG Award No: AISG2-RP-2020-016).

\bibliographystyle{plain}
\bibliography{blockchain_bib}

\appendix

\section{Omitted Proofs and Materials from Section~\ref{sec:part_1}}\label{app:omitted}

\begin{observation}\label{obs:bound}
Let $c_{\max}:=\max_{i\in N}{\{c_i\}}$ denote the maximum mining cost among all active miners and let $ \bar{c}=\frac{1}{n}\sum_{i=1}^n c_i$ denote the average mining cost. Then, the variance $\sigma_c^2:=\sum_{i=1}^n \(c_i-\bar{c}\)$ of the per unit mining costs of all active miners in equilibrium satisfies \[\sigma_c^2<c_{\max}\(\frac{n}{n-1}-c_{\max}\).\]
\end{observation}
\begin{proof}
Since $c_i<1$ for all $i\in N$ (recall that this equivalent to $c_i<v$ prior to normalization which is naturally satisfied), it holds that $\sum_{i=1}^nc_i^2<\sum_{i=1}^nc_i$. Along with the definition of $c^*$, cf. \eqref{eq:cstar}, this yields
\begin{align*}
\sigma_c^2&=\frac{1}{n-1}\sum_{i=1}^n \(c_i-\bar{c}\)^2= \frac{1}{n-1}\sum_{i=1}^nc_i^2-\frac{1}{n\(n-1\)}\(\sum_{i=1}^nc_i\)^2\\
&\le \frac{1}{n-1}\sum_{i=1}^nc_i-\frac{n-1}{n}\(\frac{1}{n-1}\sum_{i=1}^nc_i\)^2=c^*\(1-\frac{n-1}{n}c^*\).
\end{align*}
The participation constraint, $c_i<c^*$ for all $i\in N$, implies, in particular, that $c_{\max}<c^*$. Moreover, $c^*=\frac{1}{n-1}\sum_{i=1}^nc_i<\frac{n}{n-1}c_{\max}$. Substituting these in the last expression of the above inequality, we obtain that
\begin{align*}\sigma_c^2&<c^*\(1-\frac{n-1}{n}c^*\)<\frac{n}{n-1}c_{\max}\(1-\frac{n-1}{n}c_{\max}\)=c_{\max}\(\frac{n}{n-1}-c_{\max}\).\qedhere
\end{align*}
\end{proof}
To gain some intuition about the order of magnitude of the bound derived in \Cref{obs:bound} in real applications, we consider the BTC network. Currently, the cost to produce 1 TH/s consistently for a whole day is approximately equal to $\$0.08$. On the other hand, the total miners' revenue per day is in the order of magnitude of $\$10$million. Thus, in normalized units (as the ones that we work here), $c_i$ would be equal to $c_i=0.08/10m=\$8e-09$.

\begin{proof}[Proof of \Cref{thm:grief}] 
Part (i). For $\Delta<x^*$, \Cref{cor:single} implies that  
\[\Pi_j\(x^*\)-\Pi_j\(x^*+\Delta,x^*_{-i}\)>\Pi_i\(x^*\)-\Pi_i\(x_i^*+\Delta,x^*_{-i}\).\]
Since $\Pi_i\(x^*\)=\Pi_j\(x^*\)$ for all $i,j\in N$ by the symmetry assumption, $c_i=c>0$ for all $i\in N$, it follows that $\Pi_i\(x_i^*+\Delta,x^*_{-i}\)-\Pi_j\(x^*+\Delta,x^*_{-i}\)>0$ as claimed.\\
Part (ii). The own loss of miner $i$ by deviating to allocation $x_i^*+\Delta$ when all other miners use their equilibrium allocations $x_{-i}^*$ is equal to 
\begin{align*}
\Pi_i\(\vec{x}^*\)-\Pi_i\(x_i^*+\Delta,\vec{x}_{-i}^*\)&=\frac{x_i*}{X^*}-c_ix_i^*-\lt \frac{x_i^*+\Delta}{X^*+\Delta}-c_i\(x_i^*+\Delta\)\rt=\Delta\lt c_i-\frac{X_{-i}^*}{X^*\(X^*+\Delta\)}\rt.
\end{align*}
By \Cref{thm:arnosti}, $X^*=1/c^*$, and $x_i^*=\(1-c_i/c^*\)/c^*$, which implies that $X_{-i}^*=X^*-x_i^*=c_i/\(c^*\)^2$. Substituting in the right hand side of the above equality yields 
\begin{equation}\label{eq:denominator}
\Pi_i\(\vec{x}^*\)-\Pi_i\(x_i^*+\Delta,\vec{x}_{-i}^*\)=\Delta\lt c_i-\frac{c_i/\(c^*\)^2}{\(1/c^*+\Delta\)/c^*}\rt=\frac{\Delta^2 c_ic^*}{1+c^*\Delta}\,.
\end{equation}
Similarly, the loss incurred to any miner $j\neq i$ by miner $i$'s deviation is equal to 
\begin{align}\label{eq:single}
\Pi_j\(\vec{x}^*\)-\Pi_j\(x_i^*+\Delta,\vec{x}_{-i}^*\)&=\frac{x_j^*}{X^*}-c_jx_j^*-\lt \frac{x_j^*}{X^*+\Delta}-c_jx_j^*\rt=\frac{x_j^*\Delta}{X^*\(X^*+\Delta\)}\nonumber\\&=\frac{1}{c^*}\(1-\frac{c_j}{c^*}\)\cdot\frac{\Delta}{\(1/c^*+\Delta\)/c^*}=\frac{\Delta\(c^*-c_j\)}{1+c^*\Delta}\,.
\end{align} 
Since $c_j<c^*$ for all miners $j\in N$, the last expression is always positive (i.e., all miners incur a strictly positive loss). Summing over all $j\in N$ with $j\neq i$, equation \eqref{eq:single} yields 
\begin{align}\label{eq:numerator}
\sum_{j\neq i}^n\lt\Pi_j\(\vec{x}^*\)-\Pi_j\(x_i^*+\Delta,\vec{x}_{-i}^*\)\rt&=\frac{\Delta}{1+c^*\Delta}\lt\(n-1\)c^*-\sum_{j\neq i}c_j\rt\nonumber\\
&=\frac{\Delta}{1+c^*\Delta}\lt\(n-1\)c^*-\sum_{j=1}^nc_j+c_i\rt=\frac{\Delta c_i}{1+c^*\Delta},
\end{align}
where the last equality holds by definition of $c^*$, cf. \eqref{eq:cstar}. Combining equations \eqref{eq:denominator} and \eqref{eq:numerator}, we obtain 
\[\GF\(\vec{x}^*; \(x_i^*+\Delta,\vec{x}_{-i}^*\)\)=\left.\(\dfrac{\Delta c_i}{1+c^*\Delta}\)\middle/\(\dfrac{\Delta^2 c_ic^*}{1+c^*\Delta}\)\right.=\frac1{c^*\Delta },\]
which concludes the proof of part (ii).\\
Part (iii). For an allocation $\vec{y}=\(y_i\)_{i\in N}$ to be individually non-griefable it must hold that 
\[\Pi_j\(\vec{y}\)-\Pi_j\(y_i+\Delta,\vec{y}_{-i}\)<\Pi_i\(\vec{y}\)-\Pi_i\(y_i+\Delta,\vec{y}_{-i}\),\]
for all $i,j \in N$ with $i\neq j$ and for all $\Delta>0$. This yields the inequality (cf. equation \eqref{eq:single} in the proof of part (ii))
\[\frac{y_j\Delta}{Y(Y+\Delta)}<\Delta\lt c_i-\frac{Y-y_i}{Y\(Y+\Delta\)}\rt, \quad \text{for each } i,j\in N, \Delta>0,\]
which after some trivial algebra can be equivalently written as 
\[c_i\(Y+\Delta\)Y>Y+y_j-y_i, \quad \text{for each }i,j\in N, \Delta >0.\]
Since the left hand side is increasing in $\Delta$ and since the above must hold for each $\Delta>0$, it suffices to prove the inequality for $\Delta=0$ in which case it must hold with equality. This gives the condition 
\[c_iY^2=Y+y_j-y_i, \quad \text{for each }i,j\in N,\]
which can be now solved for the individually non-griefable allocation $\vec{y}=(y_i)_{i\in N}$. Summing over $j\neq i \in N$ yields
\[(n-1)c_iY^2=(n-1)Y+Y-y_i-(n-1)y_i, \quad \text{for each } i \in N,\]
or equivalently 
\begin{equation}\label{eq:yi}
y_i=Y\lt 1-\frac{n-1}{n}c_iY\rt, \quad \text{for each } i\in N.
\end{equation}
Summing equation \eqref{eq:yi} over all $i$ yields
\[Y=Y\lt n - \frac{n-1}{n}Y\sum_{i\in N} c_i\rt\]
which we can solve for $Y$ to obtain that $Y=\frac{n}{\sum_{i\in N}c_i}$. Using the notation of equation \eqref{eq:cstar}, this can be written as 
\[Y=\frac{n}{n-1}\cdot\frac{n-1}{\sum_{i\in N}c_i}=\frac{n}{n-1}\cdot \frac{1}{c^*}.\]  
Substituting back in equation \eqref{eq:yi} yields the unique allocations $y_i$ 
\begin{align*}y_i&=\frac{n}{(n-1)c^*}\lt 1-\frac{(n-1)c_i}{n}\frac{n}{(n-1)c^*}\rt=\frac{n}{n-1}\(1-c_i/c^*\)/c^*=\frac{n}{n-1}x_i^*,\end{align*}
where $x_i^*=\(1-c_i/c^*\)/c^*$ is the Nash equilibrium allocation for each $i\in N$ (cf. \Cref{thm:arnosti}). This concludes the proof of part (iii).
\end{proof}

\begin{proof}[Proof of \Cref{cor:single}]
By equations \eqref{eq:denominator} and \eqref{eq:single}, the inequality
\[\Pi_j\(\vec{x}^*\)-\Pi_j\(x_i^*+\Delta,\vec{x}_{-i}^*\)>\Pi_i\(\vec{x}^*\)-\Pi_i\(x_i^*+\Delta,\vec{x}_{-i}^*\)\]
is equivalent to 
\begin{align*} \frac{\Delta\(c^*-c_j\)}{1+c^*\Delta}>\frac{\Delta^2 c_ic^*}{1+c^*\Delta}&\iff c^*-c_j>\Delta c_ic^*\\&\iff \Delta<\frac1{c_i}\(1-\frac{c_j}{c^*}\).\end{align*}
Since $c_i<c^*$ by assumption, and $x_j^*=\frac{1}{c^*}\(1-\frac{c_j}{c^*}\)$ by \Cref{thm:arnosti}, the right hand side of the last inequality satisfies
\[\frac1{c_i}\(1-\frac{c_j}{c^*}\)>\frac1{c^*}\(1-\frac{c_j}{c^*}\)=x_j^*.\]
This implies that $\Delta<x_j^*$ is sufficient for the initial inequality to hold which concludes the proof.
\end{proof}

\begin{proof}[Proof of \Cref{prop:breakeven}] Part (i). Let $\Delta_i>0$ be such that $\Pi_j\(x_i^*+\Delta_1,\vec{x}^*_{-i}\)=0$. Then
\begin{align*}
\Pi_j\(x_i^*+\Delta_1,\vec{x}^*_{-i}\) = 0 &\implies \frac{x_j^*}{X^*+\Delta_1} - c_jx_j^* = 0 \\&\implies \frac{1}{X^*+\Delta_1} = c_j \\&\implies \Delta_1 = \frac{v}{c_j} - X^*
\end{align*} Since $X^* = \frac{1}{c^*}$ by \Cref{thm:arnosti}, it follows that 
\begin{equation}
\Delta_i = \frac{(c^* - c_i)}{(c_ic^*)}=\frac{1}{c_i}-\frac{1}{c^*}.
\end{equation}
The previous equation implies in particular that $\Delta_i<\Delta_j$ if and only if $c_i>c_j$ for any $i\neq j \in N$.\\
Part (ii). From equation \eqref{eq:numerator} in the proof of \Cref{thm:grief}, we know that the absolute losses, $L\(\Delta\)$, of the network when miner $i$ deviates to $x_i^*+\Delta$ are equal to
\[L\(\Delta\)=\sum_{j\neq i}^n \Pi_j\(x^*\) -\Pi_j\(x_i^*+\Delta,\vec{x}^*_{-i}\) = \frac{\Delta c_i}{1+c^*\Delta}.\]
Taking the derivative of the right hand side expression with respect to $\Delta$, we find that 
\[\frac{\partial}{\partial \Delta} L\(\Delta\)=\frac{\partial}{\partial\Delta}\frac{\Delta c_i}{1+c^*\Delta}=\frac{c_i}{(1+c^*\Delta)^2}>0.\]
This implies that the absolute losses of the network are increasing in $\Delta$. Thus, for $\Delta \in (0,\Delta_i]$, they are maximized at $\Delta=\Delta_i$ where they are equal to 
\[L\(\Delta_i\)=\frac{\(\frac{1}{c_i}-\frac1{c^*}\)c_i}{1+c^*\(\frac{1}{c_i}-\frac{1}{c^*}\)}=c_i\cdot\(1-c_i/c^*\)/c^*=c_ix_i^*,\]
where the last equality follows from \Cref{thm:arnosti}.
\end{proof}

\section{Omitted Proofs from Section~\ref{sec:multiple}: Proportional Response Dynamics with Quasi-CES Utilities}\label{app:proportional}

Our proof of \Cref{thm:equilibrium} consists of two parts. The first involves the derivation of a convex program that captures the market equilibrium (ME) spending of quasi-CES Fisher markets. To obtain this part, we utilize the approach of \cite{Bir11,Col17,Che18}. The second concerns the derivation of a general Mirror Descent (MD) algorithm which converges to the optimal solution of this convex program. Then, the last step is to show that the \textsf{PR-QCES} is an instantiation of this MD algorithm which concludes the proof.

\paragraph{Convex Program Framework.}
The convex optimization framework that we use to capture the ME spendings in quasi-CES FM is summarized in Figure~\ref{fig:derivation}. 
\begin{figure*}[!htb]
\tikzcdset{row sep/normal=1.3cm}
\tikzcdset{column sep/normal=2cm}
\begin{minipage}[b]{0.35\textwidth}
\centering
\adjustbox{scale=0.96,center}{%
\begin{tikzcd}
\text{(EG)} \arrow[r, blue!40, "\text{duality}" black]
& \text{(D)} \arrow[d, dashed, no head, blue!40, "q_j:=\ln{p_j}" black] \\
\text{(SH-QCES)} \arrow[u, blue!40, dashed, no head]
& \text{(TD)} \arrow[l, blue!40,"\text{duality}" black]
\end{tikzcd}
}
\end{minipage}\hspace{10pt}
\begin{minipage}[b]{0.6\textwidth}
\centering
\arrayrulecolor{blue!30}
\setlength{\tabcolsep}{12pt}
\small
\resizebox{\textwidth}{!}{
\begin{tabular}{@{}lll@{}c@{}l@{}r@{}}
\toprule
\textbf{Program} & \textbf{Description} &\multicolumn{4}{l}{\textbf{Variables}}\\
\midrule
(EG) & Eisenberg-Gale & $\x $ &\, $ =$\,\, & allocations \,\,& $i\in N, j\in M$\\
(D) & Dual & $p_j$ &\, $ =$\,\,& prices & $j\in M$\\
(TD) & Transformed dual & $q_j$ &\, $ =$\,\, & $\ln{\(p_j\)}$ & $j\in M$\\
(SH-QCES) & Shmyrev-type & $\b$ &\, $  =$\,\, &spending & $i\in N, j\in M$\\
\bottomrule
\end{tabular}}
\end{minipage}
\caption{Convex programs in the derivation of the \textsf{PR-QCES} protocol via the Mirror Descent (MD) protocol. Starting from the dual (D) of a generalized Eisenberg-Gale convex program (EG), we go to the transformed dual (TD) and by convex duality to a Shmyrev-type primal program (SH-QCES) which is hence, equivalent to the initial program (EG). The objective function of (SH-QCES) for quasi-CES utilities is 1-Bregman convex which implies convergence of the MD protocol.}
\label{fig:derivation}\vspace*{-0.1cm}
\end{figure*}
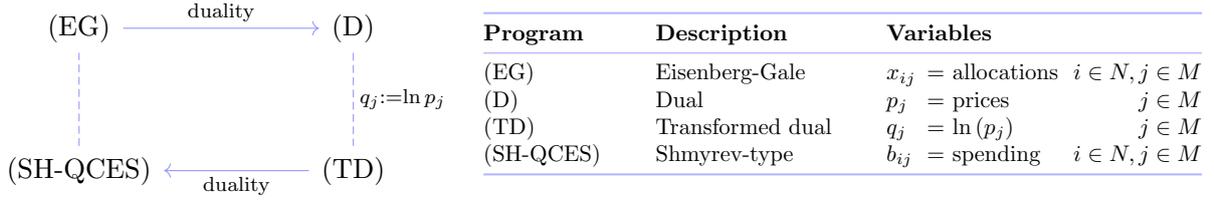
Our starting point is a convex program proposed by \cite{Dev09} that captures ME prices of quasi-linear (a sub-case of quasi-CES) Fisher markets which belongs to type (D) in Figure~\ref{fig:derivation} \cite{EG59}. From this, we derive a new convex program with captures the ME spending of the market which belongs to type (SH) \cite{Shm09}. After obtaining this new convex program, we follow the approach of \cite{Che18} to modify it so that it captures the ME spending of a quasi-CES FM. The convex program is
\begin{alignat}{6}
 \min_{\bbb,\bbw,\bbp} \quad  F(\bbb,\bbw,\bbp) \quad & 
\text{s.t.}\quad&&  \sum_{i=1}^n \b = p_j,&\,\,\forall j\in M,\nonumber\\
&&& \sum_{j=1}^m \b +w_i= K_i,&\,\, \forall i\in N,\tag{SH-QCES}\label{eq:shmyrev-quasiCES}\\
&&& \b,w_i\ge 0, &\,\, \forall i\in N, j\in M,\nonumber
\end{alignat}
where $F(\bbb,\bbw,\bbp)$ is the following function:
\begin{align*}
F(\bbb,\bbw,\bbp) := &-\sum_{i=1}^n \frac{1}{\rhoi}\sum_{j=1}^m \b\ln[\v (\b)^{\rhoi-1}]+\sum_{j=1}^m p_j\ln{p_j}+\sum_{i=1}^n \left[w_i + \frac{\rhoi-1}{\rhoi}\cdot (K_i - w_i) \ln (K_i - w_i) \right].
\end{align*}
Recall that $\b$ is the spending of agent $i$ on good $j$, $p_j := \sum_i \b$, and $w_i$ is the amount of budget/capital of agent $i$ left not spent/invested. We start by showing that the solutions of \eqref{eq:shmyrev-quasiCES} are solutions to our initial problem, i.e., to find the market equilibrium spending.

\begin{lemma}\label{lem:engineer}
Each minimum point of \eqref{eq:shmyrev-quasiCES} corresponds to a market equilibrium spending.
\end{lemma}

\begin{proof}
We verify that the optimality condition of the convex program is the same as the market equilibrium condition.
	
\paragraph{Optimality Condition.}
The partial derivatives of $F$ are
\begin{align*}
\frac{\partial}{\partial \b}F(b,w) &= \frac{1}{\rhoi} \left( 1 - \ln  \frac{\v (\b)^{\rhoi-1}}{(p_j)^\rhoi}\right)= \frac{1}{\rhoi}\left( 1 - \ln \v \right) + \frac{1-\rhoi}{\rhoi}\cdot \ln \b + \ln p_j\\
\frac{\partial}{\partial w_i} F(b,w) &= \frac{1}{\rhoi} \left[ 1 - (\rhoi-1) \ln (K_i - w_i) \right].
\end{align*}
Since $(1-\rhoi)/\rhoi > 0$, $\lim_{\b \searrow 0} \frac{1-\rhoi}{\rhoi}\cdot \ln \b = -\infty$. Hence, at each minimum point, $\b$ must be strictly positive. In turn, since $\b$ is in the relative interior of the domain at each minimum point, and we have the constraint $\sum_{j=1}^m \b \le K_i$, it must hold that all $\frac{\partial}{\partial \b}F(b,w)$ are identical for all $j$, for each buyer $i$.
Equivalently, $\frac{\v (\b)^{\rhoi-1}}{(p_j)^\rhoi}$ are identical for all $j$.
Moreover,
\begin{itemize}
\item if $K_i > w_i > 0$, then $\frac{\partial}{\partial \b}F(b,w) = \frac{\partial}{\partial w_i} F(b,w)$,\\ i.e., $\frac{\v (\b)^{\rhoi-1}}{(p_j)^\rhoi} = (K_i - w_i)^{\rhoi-1}$ for all $j$;
\item if $w_i = 0$, then $\frac{\partial}{\partial \b}F(b,w) \le \frac{\partial}{\partial w_i} F(b,w)$, i.e., $\frac{\v (\b)^{\rhoi-1}}{(p_j)^\rhoi} \ge (K_i - w_i)^{\rhoi-1}$ for all $j$.
\end{itemize}

\paragraph{Market Equilibrium Condition.}
We are interested in the rate the utility changes w.r.t. changes in spending on good $j$.
Since prices are considered as independent signals in market, 
the rate is
$\displaystyle \frac 1{p_j} \cdot \frac{\partial}{\partial x_{ij}} u_i(x_i;p) = \left( \sum_{j=1}^m \v (x_{ij})^\rhoi \right)^{1/\rhoi-1} \cdot \frac{\v (x_{ij})^{\rhoi-1}}{p_j} - 1$.
Since $\rhoi-1 < 0$ and, hence, $\lim_{x_{ij}\searrow 0} (x_{ij})^{\rhoi-1} = +\infty$, at the market equilibrium, $x_{ij}$ must be strictly positive, and hence $\b$ too.

Thus, at the market equilibrium, each $\b$ is in the relative interior of the domain, and we have the constraint $\sum_j \b \le K_i$,
it must be the case that $\frac 1{p_j} \cdot \frac{\partial}{\partial x_{ij}} u_i(x_i;p)$ are identical for all $j$. Thus,
\[\frac{\v (x_{ij})^{\rhoi-1}}{p_j} = \frac{\v (\b)^{\rhoi-1}}{(p_j)^\rhoi}\]
are identical for all $j$. We denote this (common) value by $z_i$. Then
\begin{align*}
\frac 1{p_j} \cdot \frac{\partial}{\partial x_{ij}} u_i(x_i;p) &= \left( \sum_{j=1}^m z_i x_{ij} p_j \right)^{1/\rhoi-1} \cdot z_i - 1\\&= (z_i)^{1/\rhoi} \left( \sum_{j=1}^m \b \right)^{1/\rhoi-1} - 1 \\
&= (z_i)^{1/\rhoi} (K_i -w_i)^{1/\rhoi-1} - 1.	
\end{align*}
There are two cases:
\begin{itemize}
\item If $K_i > w_i > 0$, which means $w_i$ is in the relative interior of its domain too, then the above derivative has to be zero, i.e., $z_i = (K_i - w_i)^{\rhoi-1}$ for all $i$;
\item If $w_i = 0$, then the above derivative at market equilibrium is positive or zero, i.e., $z_i \ge (K_i - w_i)^{\rhoi-1}$ for all $i$. \qedhere
\end{itemize}
\end{proof}

\paragraph{From Mirror Descent to Proportional Response.} A useful observation in \eqref{eq:shmyrev-quasiCES} is that the first and second constraints determine the values of $p_j,w_i$ in terms of $\b$'s. Thus, we may rewrite $F(\bbb,\bbw,\bbp)$ as a function of $\bbb$ only. Then the convex program has variables $\bbb$ only, and the only constraints on $\bbb$ are $\b\ge 0$ and $\sum_{j=1}^m \b \le K_i$. \par After deriving the convex program with variables $\bbb$ only, we can compute a ME spending by using standard optimization method like Mirror Descent (MD). To begin, we introduce some additional notation and recap a general result about MD~\cite{CT93,Bir11} below.\par

Let $C$ be a compact and convex set. The \emph{Bregman divergence} generated by a convex regularizer function $h$ is denoted by $d_h$,
defined as: for any $\bbb\in C, \bba \in \rint(C)$ where $\rint(C)$ is the relative interior of $C$,
\[d_h(\bbb,\bba) := h(\bbb) - \left[h(\bba) +\inner{\nabla h(\bba)}{\bbb-\bba}\right].\]

Due to convexity of the function $h$, $d_h(\bbb,\bba)$ is convex in $\bbb$, and its value is always non-negative.
The \emph{Kullback-Leibler divergence} (KL-divergence) between $\bbb$ and $\bba$ is
$\KL(\bbb \| \bba) := \sum_j b_j \cdot \ln \frac{b_j}{a_j} - \sum_j b_j + \sum_j a_j$,
which is same as the Bregman divergence $d_h$ with regularizer $h\(\bbb\) := \sum_j (b_j \cdot \ln b_j - b_j)$. For the problem of minimizing a convex function $f(\bbb)$ subject to $\bbb\in C$, the Mirror Descent (MD) method w.r.t.~Bregman divergence $d_h$ is given by the update rule in Algorithm~\ref{alg:md}.

\begin{algorithm}[!htb]
\caption{MD w.r.t. Bregman-divergence $d_h$}\label{alg:md}
\begin{algorithmic}[1]
\Procedure{MirrorDescent}{$f,C,\Gamma,d_h$}
\Initialize {$\bbb^\circ \in C$}
\While {$t>0,\bbb^t,\bbb \in C$}   
\State{$g\(\bbb,\bbb^t\)\gets \inner{\nabla f(\bbb^t)}{\bbb-\bbb^t}+d_h(\bbb,\bbb^t)/\Gamma$}
\State{$\bbb^{t+1}\la \argmin_{\bbb\in C}\{g\(\bbb,\bbb^t\)\}$}
\EndWhile
\EndProcedure
\end{algorithmic}
\end{algorithm} 

In the MD update rule, $1/\Gamma > 0$ is the step-size, which may vary with $t$ (and typically diminishes with $t$). However, in the current application of distributed dynamics, time-varying step-size and thus, update rule is \emph{undesirable} or even \emph{impracticable}, since this will require from the agents/firms to keep track with a global clock.\par
A function $f$ is \emph{$L$-Bregman convex} w.r.t. Bregman divergence $d_h$ if
for any $\bbb\in C$ and $\bba\in \rint(C)$, $f(\bba) + \inner{\nabla f(\bba)}{\bbb-\bba}\le f(\bbb) \le f(\bba) + \inner{\nabla f(\bba)}{\bbb-\bba} + L \cdot d_h(\bbb,\bba)$.


\begin{theorem}
\label{thm::plain::convex}
Suppose $f$ is an $L$-Bregman convex function w.r.t. Bregman divergence $d_h$,
and $\bbb^T$ is the point reached after $T$ applications of the mirror descent update rule in Algorithm~\ref{alg:md} with parameter $\Gamma = 1/L$. Then
\[ f(\bbb^T) - f(\bbb^*) ~\le~ L \cdot d(\bbb^*, \bbb^0)/T.\]
\end{theorem}

Using the above, we can now show that the objective function of the \eqref{eq:shmyrev-quasiCES} is a $1$-Bregman convex function w.r.t. the KL-divergence (\Cref{lem:one-Bregman-convex-quasiCES}). Subsequently, we show that \textsf{PR-QCES} can be \emph{derived} from Algorithm~\ref{alg:md} with a suitable choice of $\Gamma$. Then, \Cref{thm::plain::convex} guarantees that the updates of \textsf{PR-QCES} converge to an optimal solution of the convex program~\eqref{eq:shmyrev-quasiCES} and Theorem~\ref{thm:equilibrium} follows.

\begin{lemma}\label{lem:one-Bregman-convex-quasiCES}
The objective function $F$ of \eqref{eq:shmyrev-quasiCES} is a $1$-Breg\-man convex function w.r.t.~the divergence $\sum_{i=1}^n \frac{1}\rhoi \cdot \KL(x'_i || x_i)$.
\end{lemma}

To prove \Cref{lem:one-Bregman-convex-quasiCES}, we need the following technical result. 

\begin{lemma}\label{pr:refinement-KL}
Let $d=(d_i)_{i=1}^N$, $d'=(d'_i)_{i=1}^N$ be two vectors with non-negative entries,
and let 
 $e=(e_{11},\dots,e_{1M_1},\dots,e_{N1},\dots,e_{NM_N})$, $e'=(e'_{11},\dots,e'_{1M_1},\dots,e'_{N1},\dots,e'_{NM_N})$ be two vectors with non-negative entries, such that for each $1\le i \le N$, $\sum_{k=1}^{M_i} e_{ik} = d_i$ and $\sum_{k=1}^{M_i} e'_{ik} = d'_i$. Then $\KL\(d'\|d\) \le \KL\(e'\|e\)$.
\end{lemma}

\begin{proof}
Note that $g(q,r) = q\ln (q/r)$ is a a convex function in $q,r$ when $q,r > 0$. Thus,
\begin{align*}
\KL(d'\|d) &= \sum_{i=1}^N \(d'_i \ln \frac{d'_i}{d_i} - d'_i + d_i\)\\&= \sum_{i=1}^N M_i \cdot g \(\frac{1}{M_i}\sum_{k=1}^{M_i} e'_{ik},\frac{1}{M_i}\sum_{k=1}^{M_i} e_{ik} \)-\sum_{i=1}^N \sum_{k=1}^{M_i} e'_{ik}+\sum_{i=1}^N \sum_{k=1}^{M_i} e_{ik}\\
&\le \sum_{i=1}^N M_i \cdot \frac{1}{M_i}\sum_{k=1}^{M_i} g \left( e'_{ik},e_{ik} \)-\sum_{i=1}^N \sum_{k=1}^{M_i} e'_{ik}+\sum_{i=1}^N \sum_{k=1}^{M_i} e_{ik}\\
&= \sum_{i=1}^N \sum_{k=1}^{M_i} e'_{ik} \ln \frac{e'_{ik}}{e_{ik}}- \sum_{i=1}^N \sum_{k=1}^{M_i} e'_{ik} + \sum_{i=1}^N \sum_{k=1}^{M_i} e_{ik}\\&= \KL(e'\| e),
\end{align*}
where the inequality follows from the convexity of $g$.
\end{proof}

We can now prove \Cref{lem:one-Bregman-convex-quasiCES}. For convenience, we will use the notation $\bbz:=\(\bbb,\bbw,\bbp\)$ and  
\begin{align*}d_F\(\bbz',\bbz\)&:=F(\bbb',\bbw',\bbp') - F(\bbb,\bbw,\bbp)- \inner{\nabla F(\bbb,\bbw,\bbp)}{(\bbb'-\bbb,\bbw'-\bbw,\bbp'-\bbp)}
\end{align*}

\begin{proof}[Proof of Lemma~\ref{lem:one-Bregman-convex-quasiCES}]
We begin with the following calculations:
\begin{align}
d_F\(\bbz',\bbz\)&=-\sum_{i=1}^n \frac{1}\rhoi \sum_{j=1}^m \(\b' \ln [\v (\b')^{\rhoi-1}]- \b\ln [\v (\b)^{\rhoi-1}]\) + \sum_{j=1}^m \( p_j'\ln p_j' - p_j \ln p_j \) \nonumber\\
&\phantom{=\,\,} + \sum_{i=1}^n \Big[ (w'_i - w_i) + \frac{\rhoi-1}{\rhoi} \cdot \left[ (K_i - w'_i)\cdot\ln (K_i - w'_i) - (K_i - w_i)\ln (K_i - w_i) \right] \Big] \nonumber\\
&\phantom{=\,\,}- \sum_{i=1}^n \sum_{j=1}^m \frac{(\b' - \b)}{\rhoi} \( 1 - \ln  [\v (\b)^{\rhoi-1}]+ \rhoi \ln p_j \) \nonumber\\
&\phantom{=\,\,}- \sum_{i=1}^n \frac{(w'_i - w_i)}{\rhoi} \( 1 - (\rhoi-1)\cdot \ln (K_i - w_i) \) \label{eq:one}\\
&=-\sum_{i=1}^n \frac{\rhoi-1}{\rhoi}\sum_{j=1}^m \b' \ln \frac{\b'}{\b}  + \sum_{j=1}^m p_j'\ln \frac{p_j'}{p_j} - \sum_{i=1}^n \sum_{j=1}^m \frac{1}\rhoi \cdot (\b' - \b) \nonumber\\ 
&\phantom{=\,}+ \sum_{i=1}^n \left[ \frac{\rhoi-1}{\rhoi}  (w'_i - w_i)\frac{\rhoi-1}{\rhoi}\cdot (K_i - w'_i)\ln \frac{K_i - w'_i}{K_i - w_i}\right].\label{eq:two}
\end{align}
Let $q_i = K_i - w_i$ and $q'_i = K_i - w'_i$. Then~\eqref{eq:two} can be rewritten as 
\[\sum_{i=1}^n \frac{\rhoi-1}{\rhoi}\cdot \left(q'_i \ln \frac{q'_i}{q_i} - q'_i + q_i\right),\] which equals to $\sum_{i=1}^n\frac{\rhoi-1}{\rhoi}\cdot \KL(q_i' \| q_i)$. Recall that $q_i = \sum_{j=1}^m \b$ and $q'_i = \sum_{j=1}^m \b'$, and observe that $\frac{\rhoi-1}{\rhoi} < 0$.
By Proposition~\ref{pr:refinement-KL}, we have
\[0 \ge \sum_{i=1}^n\frac{\rhoi-1}{\rhoi}\cdot \KL(q_i' \| q_i) \ge \sum_{i=1}^n \frac{\rhoi-1}{\rhoi}\cdot \KL(x'_i \| x_i). \]
For~\eqref{eq:one}, there are two ways to rewrite it. Firstly, it can be rewritten as 
\begin{align*}\eqref{eq:one}=&-\sum_{i=1}^n \sum_{j=1}^m \b' \ln \frac{\b'}{\b}+\sum_{j=1}^m p'_j \ln \frac{p'_j}{p_j}+\sum_{i=1}^n \frac{1}\rhoi \cdot \KL(b'_i \| b_i)\\=& -\KL(b'\|b) + \KL(p'\|p)+\sum_{i=1}^n \frac{1}\rhoi \cdot \KL(b'_i \| b_i)\\\le& \sum_{i=1}^n \frac{1}\rhoi \cdot \KL(b'_i \| b_i)
\end{align*}
where the inequality holds due to Proposition~\ref{pr:refinement-KL}. Secondly, it can be rewritten as
\begin{align*}
\eqref{eq:one}=&\sum_{i=1}^n \frac{1-\rhoi}{\rhoi}\cdot \KL(b_i' \| b_i)- \sum_{i=1}^n \sum_{j=1}^m (\b' - \b) +\sum_{j=1}^m p'_j \ln \frac{p'_j}{p_j}\\
=& \sum_{i=1}^n \frac{1-\rhoi}{\rhoi}\cdot \KL(b_i' \| b_i)- \sum_{j=1}^m (p'_j - p_j)+\sum_{j=1}^m p'_j \ln \frac{p'_j}{p_j}\\
=& \sum_{i=1}^n \frac{1-\rhoi}{\rhoi}\cdot \KL(b_i' \| b_i) + \KL(p'\|p)\\
\ge& \sum_{i=1}^n \frac{1-\rhoi}{\rhoi}\cdot \KL(b_i' \| b_i).
\end{align*}
Combining all the above inequalities yields 
\[0\le d_F\(\bbz',\bbz\)\le \sum_{i=1}^n \frac{1}\rhoi \cdot \KL(b'_i \| b_i),\] as claimed.
\end{proof}

\newcommand{\of}{\overline{F}}

We now turn to the derivation of the \textsf{PR-QCES} protocol from Mirror Descent algorithm. For the convex program \eqref{eq:shmyrev-quasiCES}, the Mirror Descent rule (\Cref{alg:md}) is
\begin{align*}
(b^{t+1},w^{t+1}) = &\argmin_{(b,w)\in C} \left\{ \sum_{i=1}^n \sum_{j=1}^m \frac{(\b - \b^t)}{\rhoi}\cdot \(1 - \ln  \frac{\v (\b^t)^{\rhoi-1}}{(p_j^t)^\rhoi}\)\right.\\& \phantom{=}\left.+\sum_{i=1}^n \frac{1}{\rhoi} \left[ 1 - (\rhoi-1)\cdot \ln (K_i - w_i^t) \right]\cdot(w_i - w_i^t) +\sum_{i=1}^n \frac{1}\rhoi \cdot \KL(b_i \| b_i^t) \right\}.
\end{align*}
Since $\sum_{j=1}^m \b + w_i$ is a constant in the domain $C$, we may ignore any term that does not depend on $b$ and $w$, and any positive constant factor in the objective function and simplify the above update rule to
\begin{align*}
(b^{t+1},w^{t+1}) =& \argmin_{(b,w)\in C} \left\{ -\sum_{i=1}^n \sum_{j=1}^m \left(\ln \frac{\v (\b^t)^{\rhoi-1}}{(p_j^t)^\rhoi}\cdot \b - \b \ln \frac{\b}{\b^t}+ \b\right)+\sum_{i=1}^n (1-\rhoi) \ln (K_i - w_i^t) \cdot w_i \right\}\\
\triangleq& \argmin_{(b,w)\in C}\of(b,w).
\end{align*}
We have
\begin{align*}
\frac{\partial}{\partial \b} \of(b,w)&=-\ln \frac{\v (\b^t)^{\rhoi-1}}{(p_j^t)^\rhoi} + \ln \frac{\b}{\b^t}= \ln \b - \ln \frac{\v (\b^t)^\rhoi}{(p_j^t)^\rhoi} \\
\frac{\partial}{\partial w_i} \of(b,w)&=(1-\rhoi)\ln (K_i - w_i^t).
\end{align*}
As before, for each fixed $i$, the values of $\ln \b ~-~ \ln \frac{\v (\b^t)^\rhoi}{(p_j^t)^\rhoi}$ for all $j$ are identical. In other words, there exists $c_i > 0$ such that
\[
\b ~=~ c_i \cdot \frac{\v (\b^t)^\rhoi}{(p_j^t)^\rhoi}.
\]
As before, there are two cases which depend on $S_i \triangleq \sum_{j=1}^m \frac{\v (\b^t)^\rhoi}{(p_j^t)^\rhoi}$.
\begin{itemize}
\item If $S_i \ge K_i \cdot (K_i - w_i^t)^{\rhoi-1}$, then for each $j$ we set $\b^{t+1} = K_i \cdot \frac{\v (\b^t)^\rhoi}{(p_j^t)^\rhoi} / S_i$, and $w_i^{t+1} = 0$.
At this point, we have $\frac{\partial}{\partial \b} \of(b,w) = \ln \frac{K_i}{S_i} \le \frac{\partial}{\partial w_i} \of(b,w)$, so the optimality condition is satisfied.
\item if $S_i < K_i \cdot (K_i - w_i^t)^{\rhoi-1}$, then for each $j$, we set $\b^{t+1} = (K_i - w_i^t)^{1-\rhoi} \cdot \frac{\v (\b^t)^\rhoi}{(p_j^t)^\rhoi}$,
and $w^{t+1} = K_i - \sum_{j=1}^m \b^{t+1} > 0$.
At this point, $\frac{\partial}{\partial \b} \of(b,w) = \frac{\partial}{\partial w_i} \of(b,w)$, so the optimality condition is satisfied.
\end{itemize}

This concludes the proof of \Cref{thm:equilibrium} which shows that the \eqref{eq:proportional} dynamics converge to the market equilibrium of a Fisher market with quasi-CES utilities. The result holds for any $0<\rho_i\le 1$. However, the proof cannot be extended in a straightforward way to values of $\rho_i<0$. To see this, we rewrite~\eqref{eq:one} and~\eqref{eq:two} as:
\begin{align}
d_F\(\bbz',\bbz\) = \sum_{i=1}^n \frac{1-\rhoi}{\rhoi}\cdot \left[ \KL(b_i' \| b_i) - \KL(q_i' \| q_i)\right]+\KL(p'\|p).\label{eq:negative-rho}
\end{align}
As it happens, the RHS of~\eqref{eq:negative-rho} can be either positive or negative. By Proposition~\ref{pr:refinement-KL}, 
\[\KL(b_i' \| b_i) \ge \KL(q_i' \| q_i),\]
and there are situations where the equality holds, and thus the RHS of~\eqref{eq:negative-rho} is positive.
On the other hand, it is not hard to find two points $b'\neq b$ such that $p'=p$ and $q'_i=q_i$ for all $i$ \footnote{For instance,
consider $b',b$ such that there exists two goods $j,k$ satisfying $p_j = p_k$, but there exists $i$ such that $b_{ij}\neq b_{ik}$.
Then set $b'$ same as $b$, except that $b'_{ij} = b_{ik}$ and $b'_{ik} = b_{ij}$. A sanity check verifies $p'=p$ and $q'_i=q_i$ for all $i$.},
then the RHS of~\eqref{eq:negative-rho} is negative as $\frac{1-\rhoi}{\rhoi} < 0$.
Thus, $F$ is neither convex nor concave function.

\end{document}

%% file: header.tex
\newcommand{\la}{\leftarrow}

\newcommand{\hide}[1]{}

\DeclareMathOperator*{\argmin}{arg\,min}

\newcommand{\bba}{\mathbf{a}}
\newcommand{\bbb}{\mathbf{b}}

\newcommand{\bbp}{\mathbf{p}}

\newcommand{\bbw}{\mathbf{w}}

\newcommand{\bbz}{\mathbf{z}}

\newcommand{\inner}[2]{\left\langle #1,#2\right\rangle}

\newcommand{\rint}{\mathsf{rint}}
\newcommand{\KL}{\mathrm{KL}}

%% file: Griefing_arxiv.bbl
\begin{thebibliography}{10}

\bibitem{Alk19}
C.~Alkalay-Houlihan and N.~Shah.
\newblock {The Pure Price of Anarchy of Pool Block Withholding Attacks in
  Bitcoin Mining}.
\newblock {\em AAAI Conference on Artificial Intelligence, AAAI-19}, 33(1),
  2019.

\bibitem{Ame12}
J.A. Amegashie.
\newblock Productive versus destructive efforts in contests.
\newblock {\em European Journal of Political Economy}, 28(4):461--468, 2012.

\bibitem{Arn18}
N.~Arnosti and S.~M. Weinberg.
\newblock {Bitcoin: A Natural Oligopoly}.
\newblock In Avrim Blum, editor, {\em 10th Innovations in Theoretical Computer
  Science Conference (ITCS 2019)}, volume 124, pages 5:1--5:1, 2018.

\bibitem{Aue19}
R.~Auer.
\newblock {Beyond the Doomsday Economics of Proof-of-Work in Cryptocurrencies}.
\newblock Discussion Paper DP13506, London, Centre for Economic Policy
  Research, February 2019.

\bibitem{Be16}
I.~Bentov, A.~Gabizon, and A.~Mizrahi.
\newblock Cryptocurrencies without proof of work.
\newblock In Jeremy Clark, Sarah Meiklejohn, Peter~Y.A. Ryan, Dan Wallach,
  Michael Brenner, and Kurt Rohloff, editors, {\em Financial Cryptography and
  Data Security}, pages 142--157, Berlin, Heidelberg, 2016. \, Springer Berlin
  Heidelberg.

\bibitem{Bir11}
B.~Birnbaum, N.~R. Devanur, and L.~Xiao.
\newblock {Distributed Algorithms via Gradient Descent for Fisher Markets}.
\newblock In {\em {EC}'11}, pages 127--136. ACM, 2011.

\bibitem{Bis19}
G.~Bissias, B.~N. Levine, and D.~Thibodeau.
\newblock {Greedy but Cautious: Conditions for Miner Convergence to Resource
  Allocation Equilibrium}, 2019.

\bibitem{Bo15}
J.~Bonneau.
\newblock {Why Buy When You Can Rent?}
\newblock Financial Cryptography and Data Security, pages 19--26, Berlin,
  Heidelberg, 2016. Springer Berlin Heidelberg.

\bibitem{Bon15}
J.~{Bonneau}, A.~{Miller}, J.~{Clark}, A.~{Narayanan}, J.~A. {Kroll}, and E.~W.
  {Felten}.
\newblock {SoK: Research Perspectives and Challenges for Bitcoin and
  Cryptocurrencies}.
\newblock 2015 IEEE Symposium on Security and Privacy, pages 104--121, 2015.

\bibitem{Bra18}
S.~Br{\^{a}}nzei, R.~Mehta, and N.~Nisan.
\newblock {Universal Growth in Production Economies}.
\newblock In {\em {NeurIPS} 2018}, volume~31, pages 1973--1973, 2018.

\bibitem{Bro19}
J.~Brown-Cohen, A.~Narayanan, A.~Psomas, and S.~M. Weinberg.
\newblock {Formal Barriers to Longest-Chain Proof-of-Stake Protocols}.
\newblock In {\em Proceedings of the 2019 ACM Conference on Economics and
  Computation}, EC ’19, page 459–473, New York, NY, USA, 2019. ACM.

\bibitem{Bud18}
E.~Budish.
\newblock {The Economic Limits of Bitcoin and the Blockchain}.
\newblock Working Paper 24717, National Bureau of Economic Research, June 2018.

\bibitem{But17}
V.~Buterin.
\newblock {A griefing factor analysis model}, 2018.
\newblock
  \href{https://ethresear.ch/t/a-griefing-factor-analysis-model/2338}{ethresear.ch}
  [Online; accessed: 11-February-2021].

\bibitem{But19}
V.~{Buterin}, D.~{Reijsbergen}, S.~{Leonardos}, and G.~{Piliouras}.
\newblock Incentives in ethereum’s hybrid casper protocol.
\newblock In {\em 2019 IEEE International Conference on Blockchain and
  Cryptocurrency (ICBC)}, pages 236--244. IEEE, USA, May 2019.

\bibitem{Car16}
M.~Carlsten, H.~Kalodner, S.~M. Weinberg, and A.~Narayanan.
\newblock {On the Instability of Bitcoin Without the Block Reward}.
\newblock In {\em Proceedings of the 2016 ACM SIGSAC Conference on Computer and
  Communications Security}, CCS ’16, page 154–167. ACM, 2016.

\bibitem{CT93}
G.~Chen and M.~Teboulle.
\newblock {Convergence Analysis of a Proximal-Like Minimization Algorithm Using
  Bregman Functions}.
\newblock {\em {SIAM} J. Optim.}, 3(3):538--543, 1993.

\bibitem{Chen19}
X.~Chen, C.~Papadimitriou, and T.~Roughgarden.
\newblock {An Axiomatic Approach to Block Rewards}.
\newblock In {\em Proceedings of the 1st ACM Conference on Advances in
  Financial Technologies}, AFT '19, pages 124--131, New York, NY, USA, 2019.
  ACM.

\bibitem{Cheu19}
Y.~K. Cheung, R.~Cole, and N.~R. Devanur.
\newblock {Tatonnement beyond gross substitutes? Gradient descent to the
  rescue}.
\newblock {\em Games and Economic Behavior}, 123:295--326, 2020.

\bibitem{Che18}
Y.~K. Cheung, R.~Cole, and Y.~Tao.
\newblock {Dynamics of Distributed Updating in Fisher Markets}.
\newblock In {\em {EC}'18}, pages 351--368, New York, NY, USA, 2018. ACM.

\bibitem{Col17}
R.~Cole, N.~R. Devanur, V.~Gkatzelis, K.~Jain, T.~Mai, V.~V. Vazirani, and
  S.~Yazdanbod.
\newblock {Convex Program Duality, Fisher Markets, and Nash Social Welfare}.
\newblock In {\em {EC}'17}, pages 459--460, 2017.

\bibitem{Col16}
R.~Cole and Y.~Tao.
\newblock {Large Market Games with Near Optimal Efficiency}.
\newblock In {\em {EC}'16}, pages 791--808, New York, NY, USA, 2016. ACM.

\bibitem{Vri18}
A.~{De~Vries}.
\newblock {Bitcoin's Growing Energy Problem}.
\newblock {\em Joule}, 2(5):801--805, 2018.

\bibitem{Vri20}
A.~{De~Vries}.
\newblock {Bitcoin’s energy consumption is underestimated: A market dynamics
  approach}.
\newblock {\em Energy Research \& Social Science}, 70:101721, 2020.

\bibitem{Dev09}
N.~R. Devanur.
\newblock {Fisher Markets and Convex Programs}.
\newblock Unpublished manuscript, 2009.

\bibitem{Dim17}
Nicola Dimitri.
\newblock Bitcoin mining as a contest.
\newblock {\em Ledger}, 2:31--37, 2017.

\bibitem{Dip09}
D.~DiPalantino and M.~Vojnovic.
\newblock {Crowdsourcing and All-Pay Auctions}.
\newblock In {\em EC ’09}, pages 119--128, 2009.

\bibitem{Dvi20}
K.~Dvijotham, Y.~Rabani, and L.~J. Schulman.
\newblock {Convergence of incentive-driven dynamics in Fisher markets}.
\newblock {\em Games and Economic Behavior}, 2020.

\bibitem{Eas19}
D.~Easley, M.~O'Hara, and S.~Basu.
\newblock {From mining to markets: The evolution of bitcoin transaction fees}.
\newblock {\em Journal of Financial Economics}, 134(1):91--109, 2019.

\bibitem{EG59}
E.~Eisenberg and D.~Gale.
\newblock {Consensus of Subjective Probabilities: The Pari-Mutuel Method}.
\newblock {\em Ann.~Math.~Statist.}, 30(1):165--168, 1959.

\bibitem{Eya18}
I.~Eyal and E.~G. Sirer.
\newblock Majority is not enough: Bitcoin mining is vulnerable.
\newblock {\em Commun. ACM}, 61(7):95--102, 2018.

\bibitem{Fia19}
A.~Fiat, A.~Karlin, E.~Koutsoupias, and C.~Papadimitriou.
\newblock {Energy Equilibria in Proof-of-Work Mining}.
\newblock In {\em {EC}'19}, pages 489--502, New York, NY, USA, 2019. ACM.

\bibitem{Gan19}
N~Gandal and J.~Gans.
\newblock {More (or Less) Economic Limits of the Blockchain}.
\newblock Discussion Paper DP14154, London, Centre for Economic Policy
  Research, November 2019.

\bibitem{Gar15}
J.~Garay, A.~Kiayias, and N.~Leonardos.
\newblock {The Bitcoin Backbone Protocol: Analysis and Applications}.
\newblock In E.~Oswald and M.~Fischlin, editors, {\em Advances in Cryptology -
  EUROCRYPT 2015}, pages 281--310, Berlin, Heidelberg, 2015. Springer Berlin
  Heidelberg.

\bibitem{Gor19}
G.~Goren and A.~Spiegelman.
\newblock {Mind the Mining}.
\newblock In {\em Proceedings of the 2019 ACM Conference on Economics and
  Computation}, EC '19, pages 475--487, New York, NY, USA, 2019. ACM.

\bibitem{Heh04}
B.~Hehenkamp, W.~Leininger, and A.~Possajennikov.
\newblock {Evolutionary equilibrium in Tullock contests: spite and
  overdissipation}.
\newblock {\em European Journal of Political Economy}, 20(4):1045--1057, 2004.

\bibitem{Hor10}
J.~J. Horton and L.~B. Chilton.
\newblock {The Labor Economics of Paid Crowdsourcing}.
\newblock In {\em Proceedings of the 11th ACM Conference on Electronic
  Commerce}, EC ’10, pages 209--218, 2010.

\bibitem{Kia16}
A.~Kiayias, E.~Koutsoupias, M.~Kyropoulou, and Y.~Tselekounis.
\newblock {Blockchain Mining Games}.
\newblock In {\em Proceedings of the 2016 ACM Conference on Economics and
  Computation}, EC ’16, page 365–382, New York, NY, USA, 2016. ACM.

\bibitem{Kon00}
K.~A. Konrad.
\newblock {Sabotage in Rent-Seeking Contests}.
\newblock {\em Journal of Law, Economics, \& Organization}, 16(1):155--165,
  2000.

\bibitem{Kwo19}
Y.~Kwon, J.~Liu, M.~Kim, D.~Song, and Y.~Kim.
\newblock {Impossibility of Full Decentralization in Permissionless
  Blockchains}.
\newblock In {\em Proceedings of the 1st ACM Conference on Advances in
  Financial Technologies}, AFT '19, pages 110--123, New York, NY, USA, 2019.
  Association for Computing Machinery.

\bibitem{Leo20}
N.~Leonardos, S.~Leonardos, and G.~Piliouras.
\newblock {Oceanic Games: Centralization Risks and Incentives in Blockchain
  Mining}.
\newblock In P.~Pardalos, I.~Kotsireas, Y.~Guo, and W.~Knottenbelt, editors,
  {\em Mathematical Research for Blockchain Economy}, pages 183--199, Cham,
  2020. Springer International Publishing.

\bibitem{LLSB08}
D.~Levin, K.~LaCurts, N.~Spring, and B.~Bhattacharjee.
\newblock {Bittorrent is an Auction: Analyzing and Improving Bittorrent's
  Incentives}.
\newblock {\em SIGCOMM Comput. Commun. Rev.}, 38(4):243--254, 2008.

\bibitem{Coi20}
{M. Shen}.
\newblock {Crypto Investors Have Ignored Three Straight 51\% Attacks on ETC},
  2020.
\newblock \href{https://www.coindesk.com/crypto-51-attacks-etc}{coindesk.com}
  [Online; accessed 11-February-2021].

\bibitem{Nak08}
S.~Nakamoto.
\newblock {Bitcoin: A Peer-to-Peer Electronic Cash System}, 2008.
\newblock [Accessed: 31-08-2020].

\bibitem{Shu20}
S.~Noda, K.~Okumura, and Y.~Hashimoto.
\newblock {An Economic Analysis of Difficulty Adjustment Algorithms in
  Proof-of-Work Blockchain Systems}.
\newblock In {\em Proceedings of the 21st ACM Conference on Economics and
  Computation}, EC '20, page 611, New York, NY, USA, 2020. Association for
  Computing Machinery.

\bibitem{Puu08}
T.~Puu.
\newblock {On the stability of Cournot equilibrium when the number of
  competitors increases}.
\newblock {\em Journal of Economic Behavior \& Organization}, 66(3):445--456,
  2008.

\bibitem{Reu21}
{Reuters Staff}.
\newblock {Crypto market cap surges above \$1 trillion for first time}, 2021.
\newblock
  \href{https://www.reuters.com/article/crypto-currency-int-idUSKBN29C264}{reuters.com}
  [Online; accessed 11-February-2021].

\bibitem{Sch88}
Mark~E. Schaffer.
\newblock {Evolutionarily stable strategies for a finite population and a
  variable contest size}.
\newblock {\em Journal of Theoretical Biology}, 132(4):469--478, 1988.

\bibitem{Shm09}
V.~I. Shmyrev.
\newblock An algorithm for finding equilibrium in the linear exchange model
  with fixed budgets.
\newblock {\em Journal of Applied and Industrial Mathematics}, 3(4):505, Dec
  2009.

\bibitem{Sin20}
R.~Singh, A.~D. Dwivedi, G.~Srivastava, A.~Wiszniewska-Matyszkiel, and
  X.~Cheng.
\newblock {A game theoretic analysis of resource mining in blockchain}.
\newblock {\em Cluster Computing}, 23(3):2035--2046, Sep 2020.

\bibitem{Spi18}
A.~Spiegelman, I.~Keidar, and M.~Tennenholtz.
\newblock {Game of Coins}, 2018.

\bibitem{Sta21}
{State of the Dapps}.
\newblock {Explore Decentralized Applications}, 2021.
\newblock \href{https://www.stateofthedapps.com/}{stateofthedapps.com} [Online;
  accessed 11-February-2021].

\bibitem{Sto19}
C.~Stoll, L.~Klaaßen, and U.~Gallersdörfer.
\newblock {The Carbon Footprint of Bitcoin}.
\newblock {\em Joule}, 3(7):1647--1661, 2019.

\bibitem{Jin20}
J.~Sun, P.~Tang, and Y.~Zeng.
\newblock {Games of Miners}.
\newblock In {\em Proceedings of the 19th International Conference on
  Autonomous Agents and MultiAgent Systems}, AAMAS '20, pages 1323--1331,
  Richland, SC, 2020. International Foundation for Autonomous Agents and
  Multiagent Systems.

\bibitem{Cit21}
{Wave Financial LLC}.
\newblock {Ethereum 2.0 staking, a worthwhile investment?}, 2021.
\newblock
  \href{https://www.cityam.com/ethereum-2-0-staking-a-worthwhile-investment/}{cityam.com}
  [Online; accessed 11-February-2021].

\bibitem{Gri21}
{Wikipedia contributors}.
\newblock Griefer --- {Wikipedia}{,} the free encyclopedia.
\newblock
  \url{https://en.wikipedia.org/w/index.php?title=Griefer&oldid=1006081077},
  2021.
\newblock [Online; accessed 11-February-2021].

\bibitem{Zha11}
L.~Zhang.
\newblock {Proportional response dynamics in the Fisher market}.
\newblock {\em Theoretical Computer Science}, 412(24):2691--2698, 2011.
\newblock Selected Papers from 36th International Colloquium on Automata,
  Languages and Programming (ICALP 2009).

\end{thebibliography}
